\documentclass[11pt, a4paper]{article} 

\title{Redundancy rules for MaxSAT\thanks{An extended abstract of this article has appeared in proceedings of SAT'25 \cite{BBBL-SAT-25}.
Ilario Bonacina and Maria Luisa Bonet have been supported by grants PID2022-138506NB-C21  and PID2022-138506NB-C22 (PROOFS BEYOND) funded by AEI. Massimo Lauria has been supported by the project PRIN 2022 ``Logical Methods in Combinatorics'' N.~2022BXH4R5 of the Italian Ministry of University and Research (MIUR).}}

\author{
Ilario Bonacina
\\ \small Universitat Polit\`ecnica de Catalunya
\\ \small Barcelona, Spain
\\ \small \texttt{ilario.bonacina@upc.edu}
\and
Maria Luisa Bonet
\\ \small Universitat Polit\`ecnica de Catalunya
\\ \small Barcelona, Spain
\\ \small \texttt{bonet@cs.upc.edu}
\and
Sam Buss
\\ \small University of California
\\ \small San Diego, USA
\\ \small \texttt{sbuss@ucsd.edu}
\and
Massimo Lauria
\\ \small Sapienza Universit\`a di Roma
\\ \small Rome, Italy
\\ \small \texttt{massimo.lauria@uniroma1.it}
}

\usepackage[noadjust]{cite}
\usepackage{hyperref}
\usepackage{xcolor}
\usepackage{microtype}
\usepackage{enumitem}
\usepackage{xspace}
\usepackage{amsmath,amsthm,amssymb}
\usepackage{thmtools}

\hypersetup{
colorlinks= true,
linkcolor=blue!60!black,
citecolor=blue!60!black,
urlcolor=black,
pagebackref=true,
}
\usepackage[left=3cm, right=3cm]{geometry}

\theoremstyle{definition}
\newtheorem{defi}{Definition}[section]
\theoremstyle{plain}
\newtheorem{lem}[defi]{Lemma}
\newtheorem{obs}[defi]{Observation}
\newtheorem{thm}[defi]{Theorem}
\newtheorem{cor}[defi]{Corollary}

\newtheorem{prop}[defi]{Proposition}
\theoremstyle{remark}
\theoremstyle{rem}
\newtheorem{rem}[defi]{Remark}

\usepackage[capitalize]{cleveref}
\crefname{thm}{Theorem}{Theorems}
\Crefname{thm}{Theorem}{Theorems}
\crefname{defi}{Definiton}{Definitions}
\Crefname{defi}{Defintion}{Definitions}
\crefname{rem}{Remark}{Remarks}
\Crefname{rem}{Remark}{Remarks}
\crefname{cor}{Corollary}{Corollaries}
\Crefname{cor}{Corollary}{Corollaries}
\crefname{lem}{Lemma}{Lemmas}
\Crefname{lem}{Lemma}{Lemmas}
\crefname{fact}{Fact}{Facts}
\Crefname{fact}{Fact}{Facts}
\crefname{obs}{Observation}{Observations}
\Crefname{obs}{Observation}{Observations}
\crefname{prop}{Proposition}{Proposition}
\Crefname{prop}{Proposition}{Proposition}

\DeclareMathOperator{\dom}{dom}

\newcommand{\biglor}{\bigvee}

\newcommand{\ltrue}{1}
\newcommand{\lfalse}{0}
\newcommand{\Bool}{\{\lfalse,\ltrue\}}
\renewcommand{\models}{\vDash}
\newcommand{\olnot}[1]{\overline{#1}}

\newcommand{\cost}[1]{\mathrm{cost}(#1)}

\newcommand{\assign}[1]{\overline{#1}}

\newcommand{\LPR}{\ensuremath{\mathrm{LPR}}\xspace}
\newcommand{\RAT}{\ensuremath{\mathrm{RAT}}\xspace}
\newcommand{\SPR}{\ensuremath{\mathrm{SPR}}\xspace}
\newcommand{\PR}{\ensuremath{\mathrm{PR}}\xspace}
\newcommand{\SR}{\ensuremath{\mathrm{SR}}\xspace}

\newcommand{\CSPR}{\ensuremath{\mathrm{cost\text{-}SPR}}\xspace}
\newcommand{\CSR}{\ensuremath{\mathrm{cost\text{-}SR}}\xspace}
\newcommand{\CPR}{\ensuremath{\mathrm{cost\text{-}PR}}\xspace}

\newcommand{\CLPR}{\ensuremath{\mathrm{cost\text{-}LPR}}\xspace}

\newcommand{\CBC}{\ensuremath{\mathrm{cost\text{-}BC}}\xspace}

\newcommand{\VPB}{\ensuremath{\mathrm{veriPB}}\xspace}

\newcommand{\assignvalue}[3]{\ensuremath{\mathsf{Value}_{#1}({#2} \mapsto {#3})}}

\newcommand{\PHP}[2]{\mathsf{PHP}^{#1}_{#2}}
\newcommand{\BPHP}[2]{\mathsf{bPHP}^{#1}_{#2}}

\def\liff{\leftrightarrow}

\renewcommand{\lnot}{\overline}

\def\Var{{\mathrm{Var}}}

\def\ExtVars{\mathsf{Ext}}
\def\Inj{\mathsf{Inj}}
\def\BB{\Delta_1}
\def\PP{\Delta_2}
\newcommand{\FF}[2]{\mathsf{F}^{#1}_{#2}}

\newcommand{\smalltitle}[1]{\paragraph*{#1}}

\def\Plus{\raisebox{0.2ex}{\ensuremath{+}}\xspace}

\newcommand{\ie}{\textit{i}.\textit{e}.,\ }
\newcommand{\wrt}{w.r.t.\ }
\newcommand{\eg}{\textit{e}.\textit{g}.,\ }

\makeatletter
\def\cleartheorem#1{%
    \expandafter\let\csname#1\endcsname\relax
    \expandafter\let\csname c@#1\endcsname\relax
}
\makeatother

\DeclareMathOperator{\flip}{flip}
\DeclareMathOperator{\HD}{HammingDistance}

\begin{document}

\maketitle

\begin{abstract}
The concept of redundancy in SAT leads to more expressive and powerful proof search techniques, capable of expressing a wide range of inprocessing techniques, and gives rise to interesting hierarchies of proof systems [Heule~\textit{et.al}'20, Buss-Thapen'19].
Redundancy has also been integrated in MaxSAT
[Ihalainen~\textit{et.al}'22, Berg~\textit{et.al}'23,
Bonacina~\textit{et.al}'24].

In this paper, we define a structured hierarchy of redundancy proof systems for MaxSAT, with the goal of studying its proof complexity. We obtain MaxSAT variants of proof systems such as \SPR, \PR, \SR, and others, previously defined for SAT.

All our rules are polynomially checkable, unlike [Ihalainen~\textit{et.al}'22]. Moreover, they are simpler and weaker than [Berg~\textit{et.al}'23],  and possibly amenable to lower bounds.
This work also complements the approach of
[Bonacina~\textit{et.al}'24]. Their proof systems use different rule
sets for soft and hard clauses, while here we propose a system using
only hard clauses and blocking variables. This is easier to integrate
with current solvers and proof checkers.

We discuss the strength of the systems introduced, we show some limitations of them, and we give a short
\CSR proof that any assignment for the weak pigeonhole principle
$\PHP{m}{n}$ falsifies at least $m-n$ clauses.

We conclude discussing the integration of our rules with MaxSAT resolution proof system, which is a commonly studied proof system for MaxSAT.
\end{abstract}

\maketitle

\section{Introduction}

This paper investigates new proof systems for MaxSAT that incorporate
redundancy inferences tailored to work for MaxSAT.\@  Redundancy
inferences were introduced as extensions to SAT solvers to
allow non-implicational inferences that preserve satisfiability and
non-satisfiability. For resolution and SAT solvers, the first redundancy
inferences were based on blocked clauses (BC)~\cite{Kullmann1999GeneralizationExtended}
and Resolution Asymmetric Tautology (RAT)~\cite{JHB:inprocessing,HHW:verifying}.
Other work on redundancy reasoning includes
\cite{HHW:trimming,HHW:symmetryDRAT,HKSB:PRuning,KRPH:erDRAT,HeuleBiere:Variable};
and, of particular relevance to the present paper,
are the work of
Heule, Kiesl, and Biere~\cite{HKB.19}, and the work of Buss and Thapen~\cite{BT.21}.
Redundancy inferences formalize
``\emph{without loss of generality}'' reasoning~\cite{RebolaPardoSuda:SatPreserving}
and can substantially strengthen resolution
and, in some cases, the effectiveness of SAT solvers for hard problems
such as the pigeonhole principle (PHP)~\cite{HKSB:PRuning}.
Indeed, in their full generality, redundancy inferences allow
resolution to polynomially simulate extended resolution.

{\em MaxSAT\/} is a generalization of SAT;
it is the problem of determining a truth assignment
for a CNF formula that  minimizes
the number of falsified clauses.
Although the MaxSAT problem is inherently more difficult than SAT,
in some cases MaxSAT can be adapted to be more efficient
in practice than CDCL solvers for hard problems
such as PHP~\cite{BBIMSM:dualrailMaxSat}.
There are several approaches to MaxSAT solvers, including
MaxSAT resolution \cite{BLM.07,LarrosaHeras:MaxSATandCSPs}, core-guided MaxSAT \cite{fumalik, ansotegui2013sat, morgado2013iterative, oll, narodytska2014maximum}, and
maximum-hitting-set MaxSAT~\cite{BacchusEtAl2021MaximumSatisfiabiliy,davies2011solving, saikko2016lmhs};
the present paper discusses only MaxSAT resolution. The MaxSAT
resolution proof system was first defined by Larrosa and
Heras~\cite{LarrosaHeras:MaxSATandCSPs} and proved complete by Bonet,
Levy and Manya~\cite{BLM.07}.

We define cost preserving
redundancy rules for MaxSAT mirroring the redundancy rules for SAT, called ``\CBC'', ``\CLPR'', ``\CSPR'', ``\CPR'',
and ``\CSR'' (see~\cref{def:CSR}).  The strongest of these is ``\CSR''
based on the substitution redundancy (\SR)~\cite{BT.21}. All five of these new inferences are sound for MaxSAT
reasoning (\cref{thm:soundness}).  Furthermore, we prove that \CSPR, \CPR
and \CSR are complete for MaxSAT (\cref{thm:SPR-completeness}).\@ On the other hand,
we prove that \CLPR and \CBC are incomplete for MaxSAT (\cref{thm:CLPR-incompleteness}).\@
We illustrate the power of \CSR by giving polynomial size
proofs of the cost of the blocking-variable version
of weak pigeonhole principle $\BPHP{m}{n}$ for arbitrary
numbers $m>n$ of pigeons and holes (\cref{thm:PHP-mn}).

Ours is not the first paper bringing redundancy reasoning to the context of optimization.
For instance, the work of Ihalainen, Berg and J\"arvisalo~\cite{IBJ.22},
building on~\cite{BergJarvisalo:SRAT},
introduced versions of redundancy inferences that work
with MaxSAT.\@  
In contrast to the system CPR of~\cite{IBJ.22},
all of our ``cost-'' inferences are polynomial-checkable for validity,
and thus all give traditional Cook-Reckhow proof systems
for MaxSAT.

Another system with redundancy rules for certifying unsatisfiability and optimization
is veriPB~\cite{BogaertsEtAl2023CertifiedDominance}.
Being rooted in cutting planes, veriPB is
particularly apt at certifying optimality, and it can log the
reasoning of MaxSAT solver strategies that are way out of reach of
MaxSAT resolution~\cite{BergEtAl2023CertifiedCore}.
The propositional fragment of \VPB is as strong as Extended Resolution
and DRAT, and plausibly even stronger \cite{LeszekNeil.24}.
In contrast, our systems are plausibly weaker and simpler, but yet
strong enough to prove efficiently interesting formulas. Still our
systems might be amenable to proving lower bounds.
Indeed, our explicit goal is to study the proof complexity of
redundancy rules for MaxSAT in a similar vein of what~\cite{BT.21}
does for SAT, something that is beyond the scope of \VPB, focused on proof logging actual solvers. As a starting point, we show that ``cost-BC'' and ``cost-LPR'' are not complete (\cref{thm:CLPR-incompleteness}) and we show width lower bounds on ``cost-SPR'' and an analogue of a width lower bound for ``cost-SR'' (\cref{thm:small-SR-incompleteness}).
Finally there is another work that tries to modify the redundancy
rules to make them compatible with the MaxSAT model of soft and hard
clauses~\cite{BonacinaEtAl2024MaxsatResolution}. That work has
a similar spirit to this one, but departs in the redundancy test, which checks
for inclusion rather than for reverse unit propagation. Such choice makes
the inference weaker but preserves the number of unsatisfied clauses,
hence the rule applies directly to soft clauses.

Before proceeding with the description and analysis of our systems, we
should highlight two aspects of redundancy inference that are somehow
orthogonal to the choice of the concrete inference rule: clause
deletion and new variable introduction. The applicability of
a redundancy rule depends on the clauses present in the database, and
it seems that allowing deletion of past clauses makes the system
stronger, and indeed collapses together the power of several types of inference rules
(see~\cite{BT.21}).
Likewise the possibility of introducing new variables in redundant
clauses makes all such systems as powerful as extended
resolution~\cite{Kullmann1999GeneralizationExtended}.
Coherently with the stated goal of having systems that are simple and
amenable to proof complexity analysis, in this paper we do not allow
neither clause deletion nor new variables. 

\paragraph*{Structure of the paper}
\cref{sec:prelim} contains all the necessary preliminaries,
including notation on MaxSAT and the blocking variables encoding of
MaxSAT instances (blocking variables are also used by~\cite{IBJ.22}).
\cref{sec:redundancy-rules} introduces the redundancy
rules for MaxSAT, proves their basic properties, and defines
calculi based on those rules.
\cref{sec:snd-complet} shows
their soundness, and their completeness.
\cref{sec:lower-bounds} shows the incompleteness of \CBC/\CLPR and some limitations of \CSPR and even \CSR.
\cref{sec:examples} gives examples of applications of the
redundancy rules, including a polynomial size
proof of the optimal cost of the weak
Pigeonhole Principle and a general result about the polynomial size
provability  of minimally unsatisfiable formulas.
To deal with the incompleteness of \CBC and \CLPR,
\cref{sec:maxsat_resolution_redundancy} describes proof
systems augmenting MaxSAT resolution and the systems defined in
\cref{sec:redundancy-rules}.
\cref{sec:conclusions} gives some concluding remarks.

\section{Preliminaries}
\label{sec:prelim}

For a natural number $n$, let $[n]$ be the set $\{1,\dots,n\}$.
Sets and multi-sets are denoted with capital Roman or Greek letters.

\smalltitle{Propositional logic notation} A \emph{Boolean variable}
\(x\) takes values in \(\{0,1\}\). A \emph{literal} is either a variable
\(x\) or its negation \(\lnot{x}\). A \emph{clause} is a finite
disjunction of literals, \ie $C=\biglor_{i}\ell_{i}$.
The
empty clause is \(\bot\).
A formula in \emph{Conjunctive Normal Form} (CNF) is a conjunction of clauses
$\Gamma=\bigwedge_{j}C_{j}$. We identify a CNF with the multiset of
its clauses, and denote as $|\Gamma|$ the  number of its clauses (counted with multiplicity).
We denote as \(\Var(\Gamma)\) the set of variables in \(\Gamma\).

\smalltitle{Substitutions and assignments} A \emph{substitution} $\sigma$ for
a set of variables \(X\) is a function so that \(\sigma(x)\) is either
\(0\), \(1\) or some literal defined on \(X\). For convenience, we
extend a substitution \(\sigma\) to constants and literals, setting
\(\sigma(0)=0\), \(\sigma(1)=1\), and
\(\sigma(\lnot{x})=\lnot{\sigma(x)}\) for any variable \(x \in X\).
The composition of two substitutions $\sigma,\tau$ is the substitution
$\sigma\circ\tau$, where $\sigma\circ\tau(x)=\sigma(\tau(x))$ for
$x \in X$.
A substitution $\sigma$ is an \emph{assignment} when
\(\sigma(x) \in \{0,1,x\}\) for any \(x \in X\).
The \emph{domain} of an assignment \(\sigma\) is $\dom(\sigma)=\sigma^{-1}(\Bool)$, and
\(\sigma\) is a \emph{total assignment} over \(X\) if \(X\) is its
domain, \ie $\sigma$ maps all variables in \(X\) to Boolean values.
Given a clause $C=\biglor_{i}\ell_{i}$ and a substitution~$\sigma$,
the clause~$C$ \emph{restricted} by $\sigma$, is \(
 C\!\upharpoonright_\sigma
=\biglor_{i}\sigma(\ell_{i})\ ,
\) simplified using the usual logic rules, \ie  $D\lor \lfalse = D$, $D\lor \ltrue = \ltrue$, and $D\lor \ell\lor \ell = D\lor \ell$. Another notation for \( C\!\upharpoonright_\sigma\) is \(\sigma(C)\).
If $\sigma(C)=\ltrue$ or $\sigma(C)$ is tautological we say that $\sigma\models C$, \ie \(\sigma\) \emph{satisfies} \(C\).

The \emph{restriction} of a CNF formula $\Gamma$ by \(\sigma\), denoted
as \(\Gamma\!\upharpoonright_\sigma\), is the conjunction of all clauses~\(C\!\upharpoonright_\sigma\) where $C\in \Gamma$
and \(\sigma(C)\ne 1\). The CNF $\Gamma\!\upharpoonright_\sigma$ is also a multiset.
We say that $\sigma$ \emph{satisfies} $\Gamma$ ($\sigma\models \Gamma$) if for every $C\in \Gamma$, $\sigma\models C$, \ie
$\Gamma\!\upharpoonright_\sigma=\emptyset$.
We say that $\Gamma\models C$ if for every total assignment~$\sigma$, if
$\sigma\models \Gamma$ then $\sigma\models C$.

We identify a literal \(\ell\) with the substitution that assigns
\(\ell\) to \(\ltrue\) and leaves all other variables unassigned.
Hence we use notations like \(\Gamma\!\upharpoonright_{\ell}\). Likewise, given
a clauses \(C\) we denote as \(\lnot{C}\) the assignment that maps all
literals in \(C\) to false, and we use the notation
\(\Gamma\!\upharpoonright_{\lnot{C}}\).

\smalltitle{Unit propagation}
A \emph{unit clause} is a clause of just one literal.
Unit propagation works as follows. Start with a CNF \(\Gamma\): if
\(\Gamma\) has no unit clauses, the process ends, otherwise pick some unit clause
\(\ell\) in \(\Gamma\) arbitrarily, remove from \(\Gamma\) all clauses containing \(\ell\) and remove literal \(\bar \ell\) from all clauses containing it. Keep repeating until \(\Gamma\) has no more unit clauses.
Regardless of the choice of the unit clause, the process always ends
with the same formula.

We say that \(\Gamma \vdash_{1} C\) when the application of unit
propagation to the formula \(\Gamma\!\upharpoonright_{\lnot{C}}\) produces the empty clause.
For two CNF formulas $\Gamma, \Delta$ we say that
$\Gamma\vdash_1 \Delta$ if for every
$D\in \Delta$, \(\Gamma \vdash_{1} D\).
Clearly, if $\Gamma\supseteq \Delta$ then $\Gamma\vdash_1\Delta$, and
if $\Gamma\vdash_1\Delta$, then $\Gamma\models\Delta$. It is important
to stress that the \(\vdash_1\) relation is \emph{efficiently checkable}.

\begin{obs}[{\cite[Fact 1.3]{BT.21}}]
\label{obs:unit-prop-restriction}
	Let $\sigma$ be a substitution and $\Gamma,\Delta$ be CNF formulas, if $\Gamma\vdash_1 \Delta$, then $\Gamma\!\upharpoonright_\sigma \vdash_1 \Delta\!\upharpoonright_\sigma$.
\end{obs}

\begin{proof}
	By assumption, for every clause $D\in \Delta$, $\Gamma\!\upharpoonright_{\lnot{D}}\ \vdash_1 \bot$. Let $D=\bigvee_{i\in I}\ell_i$. Saying that $\Gamma\!\upharpoonright_{\lnot{D}}\vdash_1 \bot$ is equivalent to saying $\Gamma\land \bigwedge_{i\in I}\lnot{\ell_i}\ \vdash_1 \bot$. This in turn implies \begin{equation}
		\label{eq:unit}
	(\Gamma\land \bigwedge_{i\in I}\lnot{\ell_i})\!\upharpoonright_\sigma\ \vdash_1 \bot ,
		\end{equation}
		 since $\sigma$ maps unit literals into unit literals (or $\ltrue$, $\lfalse$).
	By basic properties of substitutions we have  $(\Gamma\land \bigwedge_{i\in I}\lnot{\ell_i})\!\upharpoonright_\sigma=\Gamma\!\upharpoonright_\sigma\land \bigwedge_{i\in I}(\lnot{\ell_i})\!\upharpoonright_\sigma$ and $(\lnot{\ell_i})\!\upharpoonright_\sigma=\overline{(\ell_i\!\upharpoonright_\sigma)}$, that is \eqref{eq:unit} implies
	\[
	\Gamma\!\upharpoonright_{\sigma}\!\upharpoonright_{\overline{D|\sigma}}\ \vdash_1\bot \ 
	\]
	and $\Gamma\!\upharpoonright_\sigma\vdash_1 D\!\upharpoonright_\sigma$.
\end{proof}

\smalltitle{Resolution} Resolution is a well-studied propositional deduction system with
two inference rules: (i) from a clause~\(A\) we can deduce any \(B\) s.t.\ \(A \subseteq B\);
(ii) from clauses \(A \lor x\) and \(B \lor \lnot{x}\) we can deduce
\(A \lor B\). A \emph{resolution proof} from a set of clauses \(\Gamma\) is
a sequence of clauses \(D_{1}, D_{2}, \ldots, D_{t}\) where each
\(D_{i}\) is either already in \(\Gamma\) or is deduced from earlier
clauses in the sequence using one of the two inference rules.
Resolution is implicationally complete, thus deciding whether a clause \(C\) can be
deduced from \(\Gamma\) is the same as deciding whether
\(\Gamma \models C\).

\smalltitle{MaxSAT} Given a CNF formula $F$, MaxSAT asks to
 find the maximum number of clauses in $F$ which can be
simultaneously satisfied.
In applications, it is useful to consider a generalization for which we divide the
clauses into \emph{hard} or \emph{soft} (\emph{partial} MaxSAT). Hard clauses must be
satisfied, while soft clauses can be falsified with a cost (a falsified clause adds one to the cost of the assignment).
Consider \(F= H \land S\) where \(H\) is the multiset of hard clauses
and \(S\) is the multiset of soft ones. In this model, MaxSAT asks to
find the maximum number of clauses in \(S\) that can be simultaneously
satisfied by an assignment that satisfies all clauses in \(H\); or equivalently, the minimum number of clauses in $S$ that must be falsified by any assignment satisfying $H$ (``cost'').
Notice that the optimization problem is non-trivial only if \(H\) is satisfiable.%
\footnote{An even more general version is \emph{weighted} MaxSAT,
  where we would consider \emph{weighted} set of clauses \((F,w)\),
  \ie each clauses \(C \in F\) has an associated weight \(w(C)\) where
  $w: F \to \mathbb N\cup\{\infty\}$. In this model the goal is to
  minimize the weight of the falsified clauses. The role of the weight
  $\infty$ is to model \emph{hard} clauses. In this paper we do not
  focus on this model.}
It is not relevant whether \(H\) is a set or a multiset. In \(S\), on the
other hand, the multiplicity of  soft clauses must be accounted for.

Proof systems for MaxSAT aim to show lower bounds on the cost of (partial)
MaxSAT instances, i.e. the minimum number of soft clauses that need to be falsified. One such system is MaxSAT resolution. 

\smalltitle{MaxSAT with blocking variables} Without loss of generality we can
assume that all soft clauses in a MaxSAT instance are unit clauses; indeed, using a new variable $b$, a soft clause $C$ can
be replaced with a hard clause
$C\lor b$ and a soft clause $\lnot{b}$, without affecting the cost.
The variable $b$ is usually called a \emph{blocking variable}. This appears in \cite{GW94}, but it might have been used even earlier.

\begin{defi}
	  Let \(F = H \land S\) with soft clauses \(S=C_{1} \land \dots \land C_{m}\).
  The \emph{blocking variables formulation} of $F$ is  \(F'= H' \land S'\) where
  \begin{itemize}
    \item
    \(H' = H \land (C_{1} \lor b_{1}) \land \dots \land (C_{m} \lor
    b_{m})\),
    \item\(S' = \lnot{b_{1}} \land \dots \land \lnot{b_{m}}\),
  \end{itemize}
  and $b_1,\dots,b_m$ are new  variables \emph{(blocking variables)} not appearing in $F$.
  We say that $\Gamma$ is a MaxSAT instance \emph{encoded with blocking variables}, when it is given as a set of hard clauses of the form as in $H'$ above. The soft clauses, then, are not part of $\Gamma$.
\end{defi}

\begin{obs}\label{obs:softliterals}
Let $F=H\land S$ be a MaxSAT instance and $F'=H'\land S'$ be the blocking variables formulation of $F$.
  Any assignment that satisfies \(H\) and falsifies \(k\) clauses in
  \(S\) can be extended to an assignment that satisfies \(H'\) and
  sets \(k\) blocking variables to true. Vice versa, any
  assignment that satisfies \(H'\) and sets \(k\) blocking variables
  to true satisfies \(H\) too and falsifies at most \(k\) clauses in
  \(S\).
\end{obs}

Because of \cref{obs:softliterals}, for the rest of this
work we consider $\Gamma$ (corresponding to $H'$) to be a MaxSAT instance encoded with blocking variables usually named $\{b_1,\dots,b_m\}$. The soft clauses of the blocking variable formulation are not part of $\Gamma$.
The goal is to satisfy \(\Gamma\) while
setting to true the least number of blocking variables. More formally,
given a total assignment \(\alpha\) for \(\Gamma\), we define
\[
  \cost{\alpha} = \sum^m_{i=1} \alpha(b_i)   \quad \hbox{and} \quad
  \cost{\Gamma} = \min_{\alpha\, :\, \alpha \,\models\, \Gamma} \cost{\alpha}
\]
and the goal is to find the value of $\cost{\Gamma}$.
Notice that, the notation $\cost{\alpha}$ is defined even for assignments not satisfying $\Gamma$.

\section{Redundancy rules for MaxSAT}
\label{sec:redundancy-rules}

In the context of SAT, a clause $C$ is redundant \wrt a CNF instance $\Gamma$ if $\Gamma$ and $\Gamma\cup\{C\}$ are equisatisfiable, that is either they both are satisfiable or both unsatisfiable \cite{Kullmann1999GeneralizationExtended}.
The natural adaptation of this notion to MaxSAT is a clause $C$ that does not affect the cost of $\Gamma$.

\begin{defi}[redundant clause, \cite{IBJ.22}]
  A clause \(C\) is \emph{redundant} \wrt a MaxSAT instance \(\Gamma\) when
  \begin{equation}
  \label{eq:cost}
    \cost{\Gamma} = \cost{\Gamma \cup \{C\}}\;.
    \end{equation}
\end{defi}
Clauses that logically follow from \(\Gamma\) are
redundant: any optimal assignment for $\Gamma$ would satisfy them as well. There may be other useful clauses that do not
follow logically, and yet do not increase the cost if added.

The condition in eq. \eqref{eq:cost} is not polynomially checkable (unless, say
$\hbox{P}=\hbox{NP}$). Therefore, we consider efficiently certifiable notions of redundancy, \ie ways to add redundant clauses (in the sense of eq.~\eqref{eq:cost}) while certifying efficiently their redundancy.
This is done showing how to extend in a systematic way the notions of
efficiently certifiable redundancy already studied in the context of
SAT (BC, RAT, LPR, SPR, PR, SR) \cite{HKB.19, BT.21} to the context of
MaxSAT\@. This is an alternative to the approach of~\cite{IBJ.22}. This definition could also be seen as a very special case of veriPB \cite{BogaertsEtAl2023CertifiedDominance} (see \cref{sec:veriPB} for more details on the connections with veriPB).

\begin{defi}
\label{def:CSR}
A clause $C$ is \emph{cost substitution redundant} (\CSR) \wrt
to  \(\Gamma\) if there exists a substitution \(\sigma\) such that
\begin{enumerate}
	\item $\Gamma \!\upharpoonright_{\assign{C}}\ \vdash_1 (\Gamma \cup \{C\})\!\!\upharpoonright _\sigma$ \hfill \emph{(redundancy)}\label{item:CSRdef1}
	\item for all total assignments $\tau\supseteq \assign{C}$,
    $\cost{\tau \circ \sigma} \leq \cost{\tau}$ \hfill\emph{(cost)}.\label{item:CSRdef2}
\end{enumerate}

If the substitution $\sigma$ has some additional structure, we have the following redundancy rules listed in decreasing order of generality:
{
  \begin{description}
    \item[Cost propagation redundant (\emph{\CPR})] if \(\sigma\) is a partial assigment.
    \item[Cost subset propagation redundant (\emph{\CSPR})] if \(\sigma\)
    is a partial assigment with the same domain as \(\lnot{C}\). In other words,
    \(\sigma\) flips some variables in \(\lnot{C}\).
    \item[Cost literal propagation redundant (\emph{\CLPR})] if \(\sigma\) is
    a partial assignment with the same domain as \(\lnot{C}\), but
    differs from~\(\lnot{C}\) on exactly one variable.
  \end{description}
  }
\end{defi}

\Cref{item:CSRdef1} in \cref{def:CSR} claims that any assignment
\(\alpha\) that satisfies \(\Gamma\) and falsifies \(C\) must satisfy
the right hand side, and therefore the assignment \(\alpha \circ \sigma\)
must satisfy \(\Gamma \cup \{C\}\). Hence adding clause \(C\) does not
make the formula unsatisfiable, unless it was already the case.
Together with \cref{item:CSRdef2}, it ensures that any assignment
that falsifies the new clause \(C\) can be patched with a substitution
\(\sigma\) so that \(C\) is satisfied without increasing the
minimum cost (see~\cref{lem:CSR-redundant}).

Indeed, \cref{item:CSRdef1} in \cref{def:CSR} and its variants correspond to the redundancy rules of proof systems \SR, \PR, \SPR, and \LPR from \cite{BT.21}, adapted here to consider cost.
Since \LPR is the same as \RAT (see~\cite[Theorem 1.10]{BT.21}), the
notion of cost literal propagation redundancy could as well be called
\emph{cost-\RAT redundancy}.

\begin{rem}\label{rmk:IBJ}
  It is important to compare \cref{def:CSR} with \cite[Definition 2]{IBJ.22}. Redundancy conditions are very similar and the main differences are in the cost conditions.
  Let us compare their \texttt{CPR} rule with our \CPR.\@
  In \CPR, the cost condition requires the witness $\sigma$ to be at least as good as all possible extensions of $\lnot{C}$, while in \texttt{CPR} the requirement is enforced only on those extensions of $\lnot{C}$ that satisfy $\Gamma$.
  The latter more permissive condition allows to derive more clauses,
  but it is unlikely to be polynomially checkable, while the condition in
  \CPR is polynomially checkable (see~\cref{lem:CSR-p-time}).
  In \cite{IBJ.22},
  the authors also define two  polynomially checkable rules where the cost condition is not present, but implicitly enforced via restrictions on the type of assignments used. Those rules are special cases of \CLPR and \CSPR respectively.
\end{rem}

The very first notion of adding systematically redundant clauses is the one of \emph{blocked clause} \cite{Kullmann1999GeneralizationExtended}. Since this notion does not match exactly the scheme of \cref{def:CSR} \Cref{item:CSRdef1}, we define  \CBC directly as follows.

\begin{defi}[\CBC]
\label{defi:BC}A clause $C\lor \ell$ is \CBC w.r.t. a set of clauses $\Gamma$ and a literal $\ell$ if for every clause $D\lor \lnot{\ell}\in \Gamma$, $C\lor D$ is a tautology,
 and $\ell$ is not a blocking variable with positive polarity.
\end{defi}

\begin{prop}
	If $C\lor \ell$ is \CBC with respect to $\Gamma$ and the literal $\ell$, then $\cost{\Gamma} = \cost{\Gamma \cup \{C\lor \ell\}}$.
 \end{prop}

\begin{proof}
It is enough to show that \(\cost{\Gamma}\geq \cost{\Gamma\cup \{C\lor \ell\}}\).
  Let \(\cost{\Gamma}=k\).
 Consider an
  optimal total assignment $\alpha$ that satisfies~\(\Gamma\) and sets to
  true exactly \(k\) blocking variables.
  If $\alpha\models C\lor \ell$ we already have that $\alpha\models \Gamma\cup \{C\lor \ell\}$ and $\cost{\alpha}=k$.
  Otherwise, $\alpha$ extends $\assign{C\lor \ell}$. Consider the assignment $\beta$ that is equal to $\alpha$ except in $\ell$ where the value is flipped.
  
  By construction $\beta\models C\lor \ell$, moreover $\beta\models \Gamma$ since every clause which was satisfied by $\alpha$ and now might not be, has the form $D\lor \lnot{\ell}$, and hence by the blocked clause definition is satisfied by $\overline{C}$ and therefore $\beta$.
  Moreover, $\cost{\beta}\leq k$ since the flipped variable does not increase the cost.
\end{proof}

\begin{obs}
Any \CBC clause $C\lor \ell$ w.r.t. to $\Gamma$ and $\ell$ is also a \CLPR clause w.r.t. $\Gamma$.
\end{obs}
\begin{proof}
  Let $\sigma$ be defined as $\assign{C}$ and $\ell=1$. By definition
  setting $\ell=1$ does not increase the cost, so the cost-condition
  is satisfied, and furthermore it satisfies $C\lor \ell$.
  We need to show that
  $\Gamma \!\upharpoonright_{\assign{C\lor \ell}}\ \vdash_1 \Gamma |
  _\sigma$. All clauses not containing the variable of \(\ell\) are in
  the sets on both sides. The one with literal \(\ell\) are satisfied
  on the right side. The only clauses we need to check are the ones of
  the form \(D \lor \lnot{\ell}\). In that case \(D \lor C\) is
  a tautology, and the assignment \(\lnot{C}\), satisfies \(D\).
  This means that these clause are all satisfied and do not occur in 
  \(\Gamma \!\upharpoonright_\sigma\).
\end{proof}

\begin{lem}\label{lem:CSR-redundant}
  If $C$ is \CSR
  \wrt to \(\Gamma\), then $C$ is redundant \wrt $\Gamma$.
\end{lem}
\begin{proof}
 It is enough to show that \(\cost{\Gamma}\geq \cost{\Gamma\cup \{C\}}\).
  Let \(\cost{\Gamma}=k\).
  To show that adding \(C\) to $\Gamma$ does not increase the cost, consider an
  optimal total assignment $\alpha$ that satisfies~\(\Gamma\) and sets to
  true exactly \(k\) blocking variables.
  If $\alpha\models C$ we already have that $\alpha\models \Gamma\cup \{C\}$ and $\cost{\alpha}=k$.
  Otherwise, $\alpha$ extends $\assign{C}$ and, by assumption, there is
  a substitution $\sigma$ such that \(\cost{\alpha\circ \sigma}\leq k\). To show that $\cost{\Gamma\cup\{C\}}\leq k$, it remains to show that indeed $\alpha\circ \sigma\models \Gamma\cup\{C\}$. By assumption,
  \[
    \Gamma\!\upharpoonright_{\assign{C}}\ \vdash_1 (\Gamma \cup \{C\})| _\sigma \;,
  \]
and, since $\alpha\models \Gamma$ and extends $\assign{C}$,
  then $\alpha\models (\Gamma\cup \{C\})| _\sigma$  too.
  Equivalently, $\alpha \circ \sigma\models \Gamma \cup \{C\}$.
\end{proof}

Both \cref{item:CSRdef1} and \cref{item:CSRdef2} of \cref{def:CSR} are stronger than what is actually needed for \cref{lem:CSR-redundant} to hold. Indeed, for  \cref{item:CSRdef1}, it would be enough that
$\Gamma \!\upharpoonright_{\assign{C}}\ \models (\Gamma \cup \{C\})| _\sigma$, and, for  \cref{item:CSRdef2},
it would be sufficient to check it for any
\(\tau \supseteq \lnot{C}\) such that  \(\tau\models \Gamma\). Unfortunately, these latter versions of \cref{item:CSRdef1} and \cref{item:CSRdef2} are in general not polynomially checkable. Instead, our conditions are checkable in
polynomial time.
\begin{lem}
\label{lem:CSR-p-time}
Let $\Gamma$ be a MaxSAT instance, $C$ a clause and $\sigma$
a substitution. There is a polynomial time algorithm to decide whether
\(C\) is \CSR \wrt \(\Gamma\), given the
substitution $\sigma$.
\end{lem}

\begin{proof}
  The redundancy condition in \cref{def:CSR} is polynomially
  checkable, since it is a unit propagation. To check the cost
  condition we need to compute
  \begin{equation}\label{eq:max-tau}
    \max_{\tau\supseteq \lnot{C}}(\cost{\tau\circ\sigma}-\cost{\tau}),
  \end{equation}
  and decide whether this value is at most zero. The value of
  \(\cost{\tau}\) is by definition the number of variables or constants in the
  sequence \(R=\langle b_{1}, b_{2}, \ldots, b_{m} \rangle\) that
  evaluate to \(1\) after applying the assignment \(\tau\).
  The value of \(\cost{\tau\circ\sigma}\) is the same, but for the
  sequence of literals (or constants)
  \(L=\langle \sigma(b_{1}), \sigma(b_{2}), \ldots, \sigma(b_{m})
  \rangle\).
  All the expresions in the sequences \(L\) and \(R\) are either
  constant values or literals, hence each of their evaluations is
  either independent from \(\tau\) or is completely determined by the
  value assigned to the occurring variable.
  It is useful to highlight how much a single variable assignment
  increases the number of \(1\)'s in \(L\) and \(R\).

  For any sequence \(E\) of expressions that can either be \(0\),
  \(1\) or some literal, we use the following notation to indicate how
  many new \(1\)'s occur in \(E\) as a consequence of assigning some
  variable \(v\) to a value.
  \begin{align*}
    \assignvalue{E}{v}{0} =\textrm{number of occurrences of \(\lnot{v}\) in \(E\)}\\
    \assignvalue{E}{v}{1} =\textrm{number of occurrences of \(v\) in \(E\)}
  \end{align*}
  And with this notation we can write
  \begin{align*}
    \cost{\tau\circ\sigma} & = |\{i\;:\; \sigma(b_{i})=1 \}| +
                               |\{i\;:\; \sigma(b_{i}) \not\in \{0,1\}
                               \textrm{\;and\;} \tau\!\circ\!\sigma(b_{i})=1 \}| \\
                           & = |\{i\;:\; \sigma(b_{i})=1 \}| +
                               \sum_{v}\assignvalue{L}{v}{\tau(v)} \\
    \cost{\tau} & = \sum_{v}\assignvalue{R}{v}{\tau(v)}
  \end{align*}
  Let us rewrite our objective function and separate the part that is
  fixed on all \(\tau \supseteq \lnot{C}\),
  \begin{multline*}
    \cost{\tau\circ\sigma} - \cost{\tau}
    =\overbrace{|\{i\;:\; \sigma(b_{i})=1 \}| +
      \!\!\!\!\sum_{v \in \dom{(\lnot{C})}} \!\!\!\!\bigl( \assignvalue{L}{v}{\tau(v)} -
           \assignvalue{R}{v}{\tau(v)}  \bigr)}^{\textrm{Fixed part \(V\)}} +\\
         +\!\!\!\! \sum_{v \not\in \dom{(\lnot{C})}} \!\!\!\!\left( \assignvalue{L}{v}{\tau(v)} -
           \assignvalue{R}{v}{\tau(v)}  \right) \;.
   \end{multline*}
   It is crucial to observe that to achieve the maximum value of the
   objective~\eqref{eq:max-tau}, the variables not in
   \(\dom(\lnot{C})\) can be set independently from each other. Hence
   \begin{multline*}
     \max_{\tau\supseteq \lnot{C}}(\cost{\tau\circ\sigma}-\cost{\tau}) = V
     + \sum_{v \not\in \dom{(\lnot{C})}} \max_{k \in \{0,1\}}\left( \assignvalue{L}{v}{k} - \assignvalue{R}{v}{k}  \right) \;.
   \end{multline*}
   In the last equation \(V\) and the summands are easily computable
   in polynomial time, given \(C\) and \(\sigma\).
\end{proof}

\cref{lem:CSR-p-time} allows the definition of proof systems that extend resolution
using cost redundancy rules in the sense of
Cook-Reckhow~\cite{CookReckhow:proofs}.

\begin{defi}[\CSR calculus]\label{def:CSR-calculus}
  The \emph{\CSR calculus} is a proof system for MaxSAT\@.
  A derivation of a clause $C$
  from a MaxSAT instance \(\Gamma\) (encoded with blocking variables) is a sequence of clauses
  \(D_{1}, D_{2}, \ldots, D_{t}\) where, \(C \in \Gamma \cup \{D_{i}\}_{i\in [t]}\),
  each \(D_{i}\) is either
  already in \(\Gamma\) or is deduced from earlier clauses in the
  sequence using a resolution rule, or $D_i$ is \CSR \wrt to
  \(\Gamma \cup \{D_{1}, \ldots, D_{i-1}\}\) with
  \(\Var(D_{i}) \subseteq \Var(\Gamma)\).\footnote{
    We only consider the case where no new variables are added via
    \CSR rules. To be coherent with~\cite{BT.21}, $\CSR$ should be
    $\CSR^-$, following the notational convention of adding an
    exponent with ``$-$'' to denote \SR, \PR and \SPR when the systems
    are not allowed to introduce new variables. We ignore that
    convention to ease an already rich notation.}
  The length of such derivation is \(t\), \ie the number of derived clauses.
   To check the validity of derivations in polynomial time, any
  application of the \CSR rule comes accompanied by the corresponding
  substitution that witnesses its soundness.
\end{defi}
If the goal is to certify that $\cost{\Gamma}\geq s$, we can
accomplish this deriving $s$ distinct unit clauses of the form
$\{b_{i_1}, \dots, b_{i_s}\}$. If the goal
is to certify that $\cost{\Gamma}= s$, we can accomplish this deriving
$s$ distinct unit clauses of the form $\{b_{i_1}, \dots, b_{i_s}\}$
together with the unit clauses
$\{\lnot{b_j} : j\notin\{i_1,\dots,i_s\}\}$.
The soundness and completeness of these systems is their capability to
prove, respectively, correct and tight lower bounds on
\(\cost{\Gamma}\) using this schema, (see \cref{thm:soundness} and
\cref{thm:SPR-completeness}).

In a similar fashion we define \CPR, \CSPR, \CLPR, and \CBC calculi.
A remarkable aspect of these calculi is that a proof must somehow
identify the blocking variables to be set to true. When there are
multiple optimal solutions, it is quite possible that none of the
\(b_{i}\) follows logically from \(\Gamma\). Nevertheless, the
redundancy rules, often used to model ``without loss of
generality'' reasoning, can reduce the solution space.

\subsection{Simulation by veriPB}
\label{sec:veriPB}
We show that \CSR and its subsystems are p-simulated by
\VPB~\cite{BogaertsEtAl2023CertifiedDominance}. This is actually
not surprising since \VPB proof steps are based on cutting
planes, which in turn easily simulates
resolution. Furthermore \VPB includes the SAT redundancy rules, and when dealing
with optimization problems, augment them with some criteria of cost preservation.
Indeed, \VPB can argue within the proof itself whether some
substitution is cost preserving. If \VPB can do that for the type of
substitutions that respect the second condition in \cref{def:CSR},
then it p-simulates \CSR, \ie  when \CSR proves that cost is al least $k$, 
\VPB proves that $\sum_{i=1}^m b_i \geq k$ with only a polynomial overhead.

Since \VPB is a rich system, we only describe the fragment needed to
simulate \CSR.\@ Consider a CNF formula \(F\) and some objective
linear function \(f\) to be minimized. 
In our simulation we do not use neither deletions nor the dominance rule from \VPB.\@ Therefore we do not need to distinguish between core and derived constraints, and we can assume the empty preorder. For more details on \VPB and its full proof format we refer to~\cite{BogaertsEtAl2023CertifiedDominance}.

A \VPB derivation starts with \(F\) encoded as linear inequalities. At each step a linear inequality
\(C\) is derived from the set \(\Gamma\) of previously derived
inequalities in one of the following ways:
\begin{enumerate}
  \item \label{item:vpbder}
   \(C\) is derived from \(\Gamma\) using some cutting planes rule;
  \item \label{item:vpbred}  
  there exists a substitution \(\sigma\) so that
  $\Gamma \cup {\assign{C}}\ \vdash (\Gamma \cup \{C\})\!\!\upharpoonright_\sigma \cup\,
  \{f\!\!\upharpoonright_\sigma \leq f\}$, where the symbol \(\vdash\) means that there is a cutting planes
proof that witnesses the implication.
\end{enumerate}
The cutting plane proof in Item \ref{item:vpbred} is included in the
\VPB proof to make it polynomially verifiable. This is not needed if
such a proof just consists in unit propagation steps or certain uses of
Boolean axioms. In the latter cases the verifier can efficiently
recover such derivation steps on its own.

We claim that this fragment of \VPB indeed p-simulates \CSR.
Before proving that we need few additional notations and facts about
\VPB.
Every time we talk of clauses in \VPB we implicitly assume that they are represented via the 
standard encoding as linear inequalities. For instance, \
\(x \lor y \lor \neg{z}\) is represented as \(x + y + \olnot{z} \geq 1\).
Furthermore the system includes the linear inequalities \(0 \leq x \leq 1\), \(1 \leq x + \olnot{x} \leq 1\) for
each variable \(x\) in the formula (the \emph{Boolean axioms}).
The encoding uses formal representation of negative literal and variables, but the intended meaning of $\Var(C)$ is the set of all the original variables, \eg $\Var(x \lor y\lor \lnot{z})=\Var(x+y+\lnot{z}\geq 1)=\{x,y,z\}$.

For a clause $C$ encoded as $\ell\geq 1$, its negation is the inequality $\ell\leq 0$. Using the Boolean axioms, cutting planes (hence \VPB) easily proves that all the literals in $C$ must be less or equal than $0$. Therefore adding $\ell \leq 0$ is as expressive as adding the partial assignment $\overline C$. For example, the formal negation of the clause $x\lor y\lor \bar z$ is $x+y+\bar z\leq 0$ and this is equivalent modulo the Boolean axioms to $\{x\leq 0,\ y\leq 0,\ \bar z\leq 0\}$ or even $\{x= 0,\ y= 0,\ z= 1\}$.

\begin{prop}
\label{prop:veriPB-simulation}
Let \(F\)  be a MaxSAT instance over \(n\) variables with blocking
  variables \(b_{1}, b_{2}, \ldots, b_{m}\).
  If there is a \(t\) step  \CSR derivation of
  a set of clauses \(\Gamma\) from it, then
  there is a \VPB proof from \(F\)  and objective
  \(f=\sum^{m}_{i=1}b_{i}\) of some \(\Gamma'\supseteq \Gamma\) in
  \(\mathrm{poly}(nt)\) steps.
\end{prop}

\begin{proof}
  We proceed by induction on \(t\). If \(t=0\) then \(\Gamma\) is
  \(F\) itself plus the Boolean axioms, which is the initial
  configuration of the \VPB proof.
  
  For \(t>0\) we can assume that the \CSR derivation has obtained \(\Gamma\) up to
  step \(t-1\) and now at step \(t\) derives a clause \(C\) using
  either resolution, weakening or the redundancy rule, according to \cref{def:CSR}.
  For the first two cases, observe that $\Gamma\cup \overline C$ is
  unsatisfiable and immediately implies contradiction by unit
  propagation, and therefore by a short cutting planes proof. Hence,
  in both cases rule \ref{item:vpbred} of \VPB proves \(C\), with
  $\sigma$ being the empty substitution.

  The last case is when \(C\) is a cost preserving redundant clause
  with respect to $\Gamma$, witnessed by substitution \(\sigma\).
  That is
  \begin{enumerate}
	\item $\Gamma \!\upharpoonright_{\assign{C}}\ \vdash_1 (\Gamma \cup \{C\})\!\upharpoonright _\sigma$, and
	\item for all total assignments $\tau\supseteq \assign{C}$,
    $\cost{\tau \circ \sigma} \leq \cost{\tau}$.\hfill(\emph{cost condition})
\end{enumerate}
  
We show that \(C\) can be derived by the
  redundancy inference rule of \(\VPB\) from $\Gamma$ using the same substitution $\sigma$, that is
  \begin{equation}\label{eq:veriPBcond}
    \Gamma \cup {\assign{C}}\
    \vdash
    (\Gamma \cup \{C\})\!\upharpoonright_\sigma \cup\, \{f\!\!\upharpoonright_\sigma \leq f\}\;,
  \end{equation}
  and the $\vdash$ is witnessed by a short cutting planes proof.
  One part of claim in~\eqref{eq:veriPBcond}, namely the implication  $\Gamma \cup {\assign{C}}\
    \vdash
    (\Gamma \cup \{C\})\!\upharpoonright_\sigma$, follows from the unit propagation
  \(\Gamma  \cup {\assign{C}} \vdash
  \Gamma\!\upharpoonright_{\assign{C}}\) and from the assumption $\Gamma \!\upharpoonright_{\assign{C}}\ \vdash_1 (\Gamma \cup \{C\})\!\upharpoonright _\sigma$.
  
  The last part of claim \eqref{eq:veriPBcond} is the implication  \(\Gamma \cup {\assign{C}} \vdash \{f\!\upharpoonright_{\sigma} \leq f\}\).
  The  goal inequality \(f\!\upharpoonright_{\sigma} \leq f\) is indeed    \[
     \sum_{i=1}^m \sigma(b_{i}) - \sum_{i=1}^m b_{i} 
     \leq 0\;.
   \]
   Modulo the assumption $\assign{C}$  this is provably equivalent to 
     \begin{equation}
     	\label{eq:tobemaximized}
      \Big(\sum_{i=1}^m \sigma(b_{i}) - \sum_{i=1}^m b_{i}\Big)\Big\rvert_{\assign{C}}
     \leq 0 \;.
      \end{equation}

  The quantity
  \(\max_{\tau\supseteq \lnot{C}}(\cost{\tau\circ\sigma}-\cost{\tau})\)
  is the maximum achievable by the LHS of the inequality in \eqref{eq:tobemaximized} over all assignments on variables in
  \(\Var{(\Gamma)}\setminus \Var{(\assign{C})}\).  
  Using the Boolean axioms $x+\bar x=1$ we  remove negative literals form  the LHS of \eqref{eq:tobemaximized}  obtaining an equivalent representation of the form
    \begin{equation}
    	\label{eq:linearform}
    c+\sum_{v \in \Var{(\Gamma)}\setminus \Var{(\assign{C})}} c_{v} \cdot
  v\ ,
      \end{equation}
where $c$ and $c_v$s are suitable integer constants.
   The maximum is
   achieved exactly when the contribution of each variable
   \(v \not\in \Var(\lnot{C})\) is maximized. Hence
  \begin{equation}\label{eq:hasbeenmaximized}
    \max_{\tau\supseteq \lnot{C}}(\cost{\tau\circ\sigma}-\cost{\tau}) =
    c + \sum_{v  \in \Var{(\Gamma)}\setminus \Var{(\assign{C})}} \max\{c_v,0\}\;.
   \end{equation}
   Simple applications of the Boolean axioms give cutting planes proof of
   \(
     c_v \cdot v \leq
     \max\{c_v,0\} 
   \) for each \(v \not\in \Var(\lnot{C})\). The latter inequalities, summed
   together, prove that the LHS of~\eqref{eq:tobemaximized} is at most
   \eqref{eq:hasbeenmaximized}.
   Observe that formula~\eqref{eq:hasbeenmaximized} has
   a constant value at most zero, by the cost condition assumption. Therefore we have produced a cutting planes proof of  the inequality in \eqref{eq:tobemaximized}.
\end{proof}

\section{Soundness and completeness}
\label{sec:snd-complet}
Throughout this section we consider MaxSAT instances encoded with blocking variables of the form $\Gamma = H \land (C_{1} \lor b_{1}) \land \dots \land (C_{m} \lor
    b_{m})$, and we assume $H$ to be satisfiable, otherwise  the optimization problem becomes trivial.

The calculi \CSR/\CPR/\CSPR are \emph{sound} and \emph{complete} (see definitions below).
 Before proving the soundness we show an auxiliary lemma, that shows that when the calculus certifies the lower bound for the optimal values, it can also certify the optimality.

\begin{lem}\label{lem:neg-bs}
Let  \(\Gamma\) be a MaxSAT instance encoded with blocking variables $b_1,\dots,b_m$, of $\cost{\Gamma}=k$, and suppose $\Gamma$ contains the unit clauses $b_{i_1},\dots,b_{i_k}$. Then \CPR can  prove $\lnot{b_j}$ for each $j\notin\{i_1,\dots,i_k\}$ in $\mathcal O(km)$ steps.
\end{lem}
\begin{proof}
  Let $\sigma$ be a total assignment satisfying $\Gamma$ with
  $\cost{\sigma}=k$, that is $\sigma$ maps all the $b_{i_\ell}$s to
  $\ltrue$ and the other blocking variables to~$\lfalse$.
We derive all the clauses \[
	C_j=\lnot{b_{i_1}}\lor \dots\lor \lnot{b_{i_k}}\lor \lnot{b_j}
  \] with $j\notin\{i_1,\dots,i_k\}$ using the \CPR rule.
  For all clauses $C_j$,  the substitution witnessing the validity of the \CPR rule  is always $\sigma$. The redundancy condition from \cref{def:CSR} is trivially true since $\Gamma$ union an arbitrary set of $C_j$s is mapped to $\ltrue$ under $\sigma$.
	The cost condition is true because for every $\tau\supseteq \assign{C_j}$, $\cost{\tau}\geq k+1$ and $\cost{\tau\circ \sigma}=\cost{\sigma}=k$.
	
	To conclude, by resolution, derive $\lnot{b_j}$ from $C_j$ and  the unit clauses $b_{i_1},\dots,b_{i_k}$.
\end{proof}

The definition of soundness for proof systems for MaxSAT encoded with  blocking variables is the following.

\begin{defi}[soundness]
Let $\Gamma$ be a MaxSAT instance encoded with  blocking variables.
    A MaxSAT proof system $P$ is \emph{sound} if whenever $P$  derives $k$ distinct blocking variables among $b_{1},\dots,b_{m}$ (as unit clauses), then any assignment satisfying $\Gamma$  must set at least $k$ of the $b_i$s to true, in other words, $\cost{\Gamma}\geq k$.
\end{defi}

\begin{thm}[soundness of \CSR]\label{thm:soundness}
  Let $\Gamma$ be a MaxSAT instance encoded with  blocking variables.
 If there is a \CSR proof of
 $k$ distinct blocking variables, then $\cost{\Gamma}\geq k$.
   If there is a \CSR proof of
 $k$ distinct blocking variables $\{b_{i_1},\dots,b_{i_k}\}$ and all the unit clauses $\lnot{b_j}$ for $j\notin\{i_1,\dots,i_k\}$, then $\cost{\Gamma}= k$.
\end{thm}

\begin{proof}
Let  \(b_{1}, \ldots, b_{m}\) be the blocking variables of $\Gamma$.
  Let \(\Gamma'\) the set of clauses in $\Gamma$ plus all the clauses
  derived in the proof of $\cost{\Gamma}\geq k$. That is  \(\Gamma'\) also contains
  contains $k$ distinct unit clauses $b_{i_1},\dots, b_{i_k}$, hence $\cost{\Gamma^\prime}\geq k$. By \cref{lem:CSR-redundant}, the cost is
  preserved along proof steps, therefore
  \(\cost{\Gamma}=\cost{\Gamma'}\geq k\).
  In the case where we have all the $\lnot{b_j}$s then $\cost{\Gamma^\prime}=k$ and therefore $\cost{\Gamma}=k$.
\end{proof}

As an immediate consequence of \cref{thm:soundness}, also all the \CPR, \CSPR, \CLPR calculi are sound.\footnote{Notice that even if we only allow applications of the resolution rule to MaxSAT instances with blocking variables, we still obtain a sound system.}

The definition of completeness for proof systems for MaxSAT encoded with  blocking variables is the following.

\begin{defi}[completeness]
Let $\Gamma$ be a MaxSAT instance encoded with  blocking variables with $\cost{\Gamma}\geq k$, that is any assignment satisfying $\Gamma$ must set at least $k$ of the blocking variables $b_1,\dots,b_m$ to true.
    Then, a MaxSAT proof system $P$ is \emph{complete} if 
    $P$ is able to derive $k$ distinct blocking variables among $b_{1},\dots,b_{m}$ (as unit clauses).
\end{defi}

\begin{thm}[completeness of \CSPR]
\label{thm:SPR-completeness}
  Let $\Gamma$ be a MaxSAT  instance encoded with blocking variables, of $\cost{\Gamma}=k$.
There is a \CSPR derivation of the unit clauses $b_{i_1},\dots, b_{i_k}$ for some distinct $k$
  blocking literals and all the $\lnot{b_j}$ for $j\notin\{i_1,\dots,i_k\}$.
\end{thm}

\begin{proof}
Let \(b_{1}, \ldots, b_{m}\) be the blocking variables of $\Gamma$.
  Take $\alpha_{\mathrm{opt}}$ to be an optimal assignment, that is
  $\alpha_{\mathrm{opt}}\models \Gamma$,
  $\cost{\alpha_{\mathrm{opt}}} = k$, and for every total assignment
  $\beta$ that satisfies~\(\Gamma\), $\cost{\beta}\geq k$.
  Without loss of generality we can assume $\alpha_{\mathrm{opt}}$
  sets variables $b_1,\dots,b_k$ to $\ltrue$ and the remaining $b_j$s to $\lfalse$.

  Given any assignment $\gamma$, let $\bar \gamma $ be clause which is the negation of $\gamma$. In other words, $\bar \gamma$ is the largest clause
  falsified by $\gamma$. Let $\Sigma$ be the set of all clauses
  $\bar \gamma $ where $\gamma$ is a total assignment that satisfies
  $\Gamma$ and is different from $\alpha_{\mathrm{opt}}$.
  We want to derive \(\Sigma\) from \(\Gamma\), essentially forbidding
  any satisfying assignment except for \(\alpha_{\mathrm{opt}}\).

  We can add all clauses in $\Sigma$ one by one by the \CSPR rule.
  Indeed, for any clause $\bar \gamma \in \Sigma$ and any
  $\Sigma'\subseteq \Sigma$, the clause $\bar \gamma $ is \CSPR \wrt
  $\Gamma \cup \Sigma'$, with \(\alpha_{\mathrm{opt}}\) as the
  witnessing assignment. Since  $\bar{\gamma}$ has the same domain as $\alpha_\mathrm{opt}$, this is \CSPR and not just \CSR. 
  The redundancy condition
  \[
    (\Gamma \cup \Sigma')\!\upharpoonright_{\gamma}\vdash_1 (\Gamma \cup \Sigma'\cup
    \bar \gamma )\!\upharpoonright_{\alpha_{\mathrm{opt}}}
  \]
  holds because
  $\alpha_{\mathrm{opt}}\models \Gamma \cup\Sigma'\cup \bar \gamma $ so
  the RHS is just true.
  The cost condition holds by optimality of $\alpha_{\mathrm{opt}}$.
  In the end, the only assignment that satisfies \(\Gamma \cup \Sigma\)
  is $\alpha_{\mathrm{opt}}$. By the completeness of resolution we can prove each literal of $\alpha_\mathrm{opt}$, in particular literals \(b_{i}\) for
  \(1 \leq i \leq k\) and literals \(\lnot{b_{i}}\) for \(i>k\).
\end{proof}

\section{Incompleteness and Width Lower bounds}
\label{sec:lower-bounds}
\cref{thm:SPR-completeness} shows the completeness of \CSPR,
hence for \CPR and \CSR. 
It is not a coincidence that the proof does not apply to \CLPR, indeed
the latter is not complete. We will see that for some redundant
clause derived according to~\cref{def:CSR}, the number of values that
partial assignment \(\sigma\) flips with respect to any
\(\tau \supseteq \assign{C}\) may have to be large. In a \CLPR proof
this number is always at most one.
For a redundant clause \(C\) derived in some subsystem of \CSR via
a witness substitution \(\sigma\) we define
\begin{equation*}
    \flip(C,\sigma) = \max_{\tau \supseteq \assign{C}} \;\HD(\tau,\tau \circ \sigma)\;,
\end{equation*}
where \(\HD\) is the number of different bits between two total assignments.
The following result shows sufficient conditions for $\flip(C,\sigma)$ to be large.

\begin{thm}
\label{thm:flip-bound}
 Let $\Gamma$ be a MaxSAT instance encoded with blocking variables, of
  $\cost{\Gamma}=k$, and let $A$ be the set of optimal total assignments
  for \(\Gamma\), \ie $\alpha \in A$ when $\alpha \models \Gamma$ and
  $\cost{\alpha}=k$. If $A$ is such that
  \begin{enumerate}
	\item all pairs of assignments in \(A\) have Hamming distance at least $d$, and
	\item for every blocking variable $b$ there are $\alpha,\beta\in
    A$ s.t. $\alpha(b)=\lfalse$ and $\beta(b)=\ltrue$,
  \end{enumerate}
  then 
  any derivation of a blocking literal must contain a \SR step deriving a clause $C$ such that
 \(\flip(C,\sigma) \geq d\), where
  \(\sigma\) is the witnessing substitution for \(C\).
\end{thm}
\begin{proof}
 Consider a \CSR derivation from $\Gamma$ as a sequence
  $\Gamma_{0}, \Gamma_{1}, \dots, \Gamma_s$ where each
  \(\Gamma_{i+1}:=\Gamma_{i} \cup \{C\}\) with \(C\) either derived
  by resolution from clauses in \(\Gamma_{i}\), or $C$ is \CSR \wrt
  \(\Gamma_{i}\).
  For $0 \leq j \leq s $, let $\mu(j)$ be the number of the optimal
  assignments for \(\Gamma_{j}\).

  At the beginning \(\mu(0)=|A|\) by construction. If at some point
  \(\Gamma_{j}\) contains some clause~\(b_{i}\), then the value
  \(\mu(j)\) must be strictly smaller than \(|A|\) because \(A\)
  contains at least some assignment with
  \(\{b_{i} \mapsto \lfalse\}\) (by assumption 2).
  Let \(j\) be the first step where $\mu(j)$ drops below $|A|$. 
  The clause $C$ introduced at that moment must be \CSR \wrt
  $\Gamma_{j-1}$, because the resolution steps do not change the set
  of optimal assignments. Let then $\sigma$ be the witnessing
  substitution used to derive \(C\), we have
  \begin{equation}\label{eq:pi_j}
	\Gamma_{j-1}\!\upharpoonright_{\assign{C}}\ \ \vdash_1 (\Gamma_{j-1}\cup \{C\})\!\upharpoonright_\sigma\ .
  \end{equation}

  Since $\mu(j)$ dropped below \(|A|\), clause $C$ must be
  incompatible with some $\tau \in A$, that is
  $\tau \supseteq \lnot{C}$. Therefore, by cost preservation,
  \begin{equation}\label{eq:cost-k}
    \cost{\tau \circ \sigma} \leq \cost{\tau}=k\;.
  \end{equation}
  
  By \cref{obs:unit-prop-restriction}, eq.~\eqref{eq:pi_j} implies that
  \[
    \Gamma_{j-1}\!\upharpoonright_\tau\ \  \vdash_1 (\Gamma_{j-1}\cup \{C\})\!\upharpoonright_{\tau \circ \sigma}\;.
  \]
  Since $j$ was the first moment when $\mu(j)<|A|$ we have that
  $\tau\models \Gamma_{j-1}$ and therefore
  $\tau \circ \sigma \models \Gamma_{j-1}\cup \{C\}$. In particular,
  $\tau \circ \sigma\models \Gamma$.
  By eq. \eqref{eq:cost-k} then
  it must be $\tau\circ \sigma\in A$. But then, by assumption 1,
  $\tau$ and $\tau\circ \sigma$ have Hamming distance at least
  $d$, which implies \(\flip(C,\sigma)\geq d\).
\end{proof}

We see an example application of this result to \CLPR, where the
number of values allowed to be flipped by the rule is at most \(1\). 

The formula
$F=\{x\lor y \lor b_{1}, \lnot{x} \lor b_{2}, \lnot{y} \lor b_{3}\}$ has
cost \(1\) and its optimal assignments to variables
\(x,y,b_{1},b_{2},b_{3}\) are \(\{00100,10010,01001\}\).
These assignments fulfil the premises of \cref{thm:flip-bound}, with Hamming distance \(> 1\). Therefore \CLPR cannot prove the cost of $F$ to
be \(1\), and hence \CLPR is incomplete.

\begin{cor}
  \label{thm:CLPR-incompleteness}
  Proof systems \CLPR, \CBC are incomplete.\footnote{As a consequence, if we only allow applications of the resolution rule to MaxSAT instances with blocking variables, we also have an incomplete system.}
\end{cor}

\begin{cor}
  \label{thm:small-SR-incompleteness}
  There is a formula family \(F_{n}\) with \(O(n) \) variables,
  \(O(n)\) clauses and \(\cost{F_{n}}=\Omega(n)\) where, in order to
  prove \(\cost{F_{n}}\geq 1\), any \CSR proof derives
  a redundant clause with \(\flip(C,\sigma) = \Omega(n)\), where \(\sigma\) is
  the witnessing substitution for \(C\).
\end{cor}
\begin{proof}
  We define \(F_{n}\) on variables \(x_{0},x_{1},\ldots,x_{n}\), and
  variables \(y_{0},y_{1},\ldots,y_{n}\). The formula contains hard clauses
  to encode the constraint \(x_{0}\neq y_{0}\), and constraints
  \(x_{0}=x_{i}\), \(y_{0}=y_{i}\) for \(i \in [n]\).
  Furthermore \(F_{n}\) has the soft clauses, encoded as hard clauses
  with blocking variables, \(\neg x_{i} \lor b_{i}\) and
  \(\neg y_{i} \lor b_{i+n}\) for \(i \in [n]\).

  To satisfy the formula an assignment must either set all \(x\)'s to
  true and \(y\)'s to false, or vice versa. Both such assignments set
  to true \(n\) of the blocking variables, and no blocking variable is
  fixed to a constant value. Therefore the claim follows
  from~\cref{thm:flip-bound}.
\end{proof}

\begin{cor}
Any \CSPR proof for a formula respecting the hypothesis
of~\cref{thm:flip-bound} requires some redundant clauses of width at
least \(d\). 
\end{cor}
\begin{proof}
Any redundant clause $C$ derived in \CSPR must have $|C|\geq \flip(C,\sigma)$, and by \cref{thm:flip-bound}, $\flip(C,\sigma)\geq d$. 
\end{proof}
For instance, the formula $F_n$ from \cref{thm:small-SR-incompleteness} requires \CSPR refutations of width~$\Omega(n)$.

\section{Short proofs using redundancy rules}
\label{sec:examples}

We show applications demonstarting the power of the redundancy rules on notable families of CNF formulas.
In~\cref{sec:min-unsat} we consider minimally unsatisfiable formulas, while in~\cref{sec:wphp} we consider the \emph{weak Pigeonhole Principle}.

\begin{rem}
Due to \cref{thm:soundness} and \cref{thm:SPR-completeness}, we refer to a \CSPR (resp. \CPR, \CSR) derivation from $\Gamma$ of  $b_{i_1},\dots, b_{i_k}$ for some distinct $k$
  blocking literals and all the $\lnot{b_j}$ for $j\notin\{i_1,\dots,i_k\}$, as a \emph{proof of $\cost{\Gamma}=k$} in \CSPR (resp. \CPR, \CSR).
   \end{rem}

\subsection{Short proofs of minimally unsatisfiable formulas}
\label{sec:min-unsat}

Recall that \PR is the system obtained as in \cref{def:CSR} where we omit the treatment of cost, and $\sigma$ is a partial assignment (see~\cite{HKB:NoNewVariables} and~\cite[Definition 1.16]{BT.21}).
In other words, for clarity, a \PR calculus refutation of a CNF formula $F$ is a sequence $D_1,\dots,D_t$ where $D_t=\bot$, where
each $D_{i+1}$ is either in $F$, or derived by
resolution, or satisfies
\[
  \Gamma_i \!\upharpoonright_{\assign{D_{i+1}}}\ \vdash_1 (\Gamma_i \cup \{D_{i+1}\})\!\upharpoonright _\sigma\ ,
\]
where $\Gamma_i=F\cup \{D_1,\dots,D_{i}\}$ and $\sigma$
is a partial assignment.
Recall that an unsatisfiable set of clauses
is \emph{minimally unsatisfiable} if all proper subsets are satisfiable.

\begin{thm}\label{thm:PR}
  If a minimally unsatisfiable CNF formula $\{C_1,\dots,C_m\}$ has a \PR
  refutation of $s$ clauses, then there is a $\CPR$ proof of
  $\cost{\{C_1\lor b_1,\dots, C_m\lor b_m\}}= 1$ of at most
  $\mathcal{O}(s + m)$ many clauses.
\end{thm}
\begin{proof}
  Let \(F=\{C_1,\dots,C_m\}\) and
  $\Gamma=\{C_1\lor\penalty10000 b_1,\penalty1000 \dots,\penalty1000 C_m\lor\penalty10000 b_m\}$ 
  be the corresponding MaxSAT instance.
Let $\pi=(D_1,\dots,D_s)$ be a \PR refutation of $F$, so $D_s=\bot$. First we show that \[
	\pi_B=(D_1\lor B,\dots,D_s\lor B) \;,
  \]
  with $B=\biglor_{i\in [m]}b_i$, is a valid \CPR derivation of $B$
  from $\Gamma$. In particular, assuming we already derived the first
  \(i\) steps of \(\pi_{B}\), we show how to derive \(D_{i+1} \lor B \).

  When \(D_{i+i} \in F\), the clause \(D_{i+i} \lor B\) is the weakening
  of some clause in \(\Gamma\).
  If $D_{i+1}$ was derived using a resolution rule on some premises in
  \(\pi\), then \(D_{i+i} \lor B\) can be derived in the same way from
  the corresponding premises in \(\pi_{B}\).
  The remaining case is when $D_{i+1}$ is \PR \wrt   $F_{i} = F \cup \{D_{1},\ldots,D_{i}\}$. Let \(\alpha\) be the assignment that witnesses it.
  This assignment only maps variables from the original formula \(F\),
  so we extend it to
  $\alpha^\prime=\alpha \cup \{b_{1} \mapsto 0, \ldots, b_{m} \mapsto
  0\}$, and then use $\alpha^\prime$ to witness that indeed
  $D_{i+1}\lor B$ is \CPR \wrt $\Gamma_{i}=\Gamma \cup \{D_{1}\lor B, \ldots D_{i} \lor B \}$.
  For the cost condition in \cref{def:CSR}, just observe that any extension of
  $\alpha^\prime$ has cost \(0\).
  For the redundancy condition, observe that, by construction,
  \(\Gamma_{i}\!\upharpoonright_{\lnot{D_{i+1}}\wedge\lnot{B}} =
  F_{i}\!\upharpoonright_{\lnot{D_{i+1}}}\),
  \( (F_{i} \cup \{D_{i+1}\})\!\upharpoonright_{\alpha} = (\Gamma_{i} \cup \{D_{i+1} \lor B \})\!\upharpoonright_{\alpha^{\prime}} \), and $F_{i}\!\upharpoonright_{\lnot{D_{i+1}}}\vdash_1 (F_{i} \cup \{D_{i+1}\})\!\upharpoonright_{\alpha}$.

The last clause of $\pi_B$ is $B$.
  Let $\alpha_{opt}$ be an optimal assignment of \(\Gamma\). Since $F$
  is minimally unsatisfiable, $\cost{\alpha_{opt}}=1$. W.l.o.g.  assume $\alpha_{opt}$ sets \(b_{m}=1\) and all \(b_{i}=0\) for
  \(i<m\).

  Now, for each $i<m$, the clause $E_{i}=\lnot{b_i}\lor b_m$ is \CPR \wrt
  $\pi_B\cup\{E_j : j<i\}$, using $\alpha_{opt}$ itself as the
  witnessing assignment: redundancy holds since $\alpha_{opt}$
  satisfies every clause in $\pi_B$ and all clauses \(E_{j}\).
  The cost condition follow since $\cost{\tau}\geq 1$ for any
  $\tau\supseteq \assign{E_i}$ and
  $\cost{\tau\circ\alpha_{opt}}=\cost{\alpha_{opt}}=1$.

  In the end we use \(O(m)\) steps to derive \(b_{m}\) from $B$ and
  $E_1,\dots, E_{m-1}$, and to derive in \CPR calculus all the units
  $\lnot{b_1}, \dots, \lnot{b_{m-1}}$ via \cref{lem:neg-bs}.
\end{proof}

\cref{thm:PR} shows that the propositional refutations for the minimally
unsatisfiable formulas in \cite{BT.21} translate immediately to
certificates of the cost being $1$ in the MaxSAT setting. In particular, as a corollary of
\cref{thm:PR}, we have that \CPR proves in
polynomial size that
\begin{itemize}
	\item the \emph{Pigeonhole Principle} with $n+1$ pigeons and $n$ holes \cite[Theorem 4.3]{BT.21} and \cite[Section~5]{HKB.19},
	\item the \emph{Bit-Pigeonhole Principle} \cite[Theorem 4.4]{BT.21},
	\item the \emph{Parity Principle} \cite[Theorem 4.6]{BT.21},
	\item the \emph{Tseitin Principle} on a connected graph \cite[Theorem 4.10]{BT.21},
\end{itemize}
have all cost 1, since they are all minimally unsatisfiable. Since MaxSAT resolution proofs are also resolution proofs~\cite{BLM.07} and all such formulas are exponentially hard for resolution~\cite{H.85,DGM.20,Itsykson2016,U.87}, this shows a separation.

\subsection{Short proofs of the minimum cost of $\PHP{m}{n}$}
\label{sec:wphp}
Let $m > n \geq 1$. The pigeonhole principle from \(m\) pigeons to \(n\) holes, with
blocking variables, has the following formulation, that we call \(\BPHP{m}{n}\): the \emph{totality} clauses
$ \bigvee_{j\in [n]}p_{i,j}  \lor b_i $
   for $i \in [m]$, and the \emph{injectivity} clauses
$ \overline p_{i,j} \lor \overline p_{k,j} \lor b_{i,k,j}$
   for $1\le i < k \le m$ and $j\in [n]$.
We use $b_{k,i,j}$ as an alias of the variable $b_{i,k,j}$, given that \(i<k\).

\begin{thm}\label{thm:PHP-mn}
	\CSR proves $\cost{\BPHP{m}{n}}= m-n$ in polynomial size.
\end{thm}

This is the main result of the section. Before proving it we show two useful
lemmas. The first lemma is used to ``clean up'' the set of clauses during
a derivation. For each new step in a \CSR calculus derivation the
redundancy condition must be checked against an ever increasing set
of clauses. It turns out that some already derived clauses can be completely ignored for the rest of the derivation under some technical  conditions. This makes up for the lack of a deletion rule, that we do not have, and in the context of SAT seems to give more power to the systems \cite{BT.21}.

\begin{lem}\label{lmm:cleanup}
  Let \(\Gamma\) and \(\Sigma\) be two sets of clauses.
  Any \CSR derivation \(D_{1} \ldots, D_{t}\) from \(\Gamma\) is also
  a valid derivation from \(\Gamma \cup \Sigma\) if either of the two
  cases applies
  \begin{enumerate}
    \item Variables in \(\Sigma\) do not occur in \(\Gamma \cup\{D_{1},\ldots,D_{t}\}\).
    \item For every clause \(C \in \Sigma\) there is a clause
    \(C' \in \Gamma\) so that \(C' \subseteq C\).
  \end{enumerate}
\end{lem}

\begin{proof}
  The cost condition does not depend on the set of clauses, therefore
  we only need to check the validity of the redundancy condition.
  In the first case, the redundancy condition applies because the clauses of \(\Sigma\) are unaffected by
  the substitutions involved.

  For the second case, consider the derivation of a clause \(D_{i}\)
  witnessed by \(\sigma_{i}\). The clauses 
  \(\Sigma\!\upharpoonright_{\sigma_{i}}\) are weakening of clauses in \(\Gamma\!\upharpoonright_{\sigma_{i}}\). Hence
  \[
    (\Gamma\cup\{D_{1},\ldots,D_{i-1}\})\!\!\upharpoonright_{\lnot{D_{i}}}\ \vdash_{1}
    (\Gamma\cup\Sigma\cup\{D_{1},\ldots,D_{i}\})\!\!\upharpoonright_{\sigma_{i}}
  \]
  which implies the validity of the desired redundancy step
  \[
    (\Gamma\cup\Sigma\cup\{D_{1},\ldots,D_{i-1}\})\!\!\upharpoonright_{\lnot{D_{i}}}\ \vdash_{1}
    (\Gamma\cup\Sigma\cup\{D_{1},\ldots,D_{i}\})\!\!\upharpoonright_{\sigma_{i}}\ . \qedhere
  \]
\end{proof}

The second lemma is used as a general condition to enforce clauses to be \CSR.

\begin{lem}\label{lem:CSR-symm}
	Let $C$ be a clause and $\Gamma$ a set of clauses. If there exists a permutation $\pi$ such that
	\begin{enumerate}
		\item \label{item:LemSym1}$\pi$ maps the set of blocking variables to itself,
		\item \label{item:LemSym2}the substitution $\assign{C}\circ \pi$ satisfies $C$, and
 \label{item:LemSym3}
	$\Gamma\!\upharpoonright_{\assign{C}} \supseteq \Gamma\!\upharpoonright_{\assign{C}\,\circ\, \pi}$,
	\end{enumerate}
	then $C$ is \CSR w.r.t. $\Gamma$. Notice that the second condition in item \eqref{item:LemSym3} 
    is automatically satisfied if $\pi$ is a symmetry of $\Gamma$, i.e. $\Gamma = \Gamma\!\upharpoonright_\pi$.
\end{lem}

\begin{proof}
The cost condition follows
  from \cref{item:LemSym1}.
  The redundancy condition is immediate by
  \cref{item:LemSym2} using as $\sigma$
  the substitution $\overline C\circ \pi$.
\end{proof}

To satisfy
\(\Gamma\) with a clause \(C \lor b\), where \(b\) is the
corresponding blocking variable, \(b\) must be true whenever \(C\) is
false. To minimize cost, though, it makes sense to set \(b\) to
false whenever \(C\) is satisfied. Namely, to have
\(b \leftrightarrow \lnot{C}\). This does not follow logically from
\(\Gamma\), but can be derived in \(\CLPR\).

\begin{lem}\label{lem:extension_vars}
  Let \(\Gamma\) contain a clause $\{C \lor b\}$, so that \(b\) is
  a blocking variables and that clause is its unique occurrence in
  \(\Gamma\).
  It is possible to introduce all the clauses of the form
  $\lnot{\ell}\lor \lnot{b}$, for every literal $\ell$ in $C$, using the \CLPR rule.
  That is, we can turn $b$ into a full-fledged extension variable such that $b \leftrightarrow \lnot{C}$.
\end{lem}

\begin{proof}
  Let $\Gamma$ be as in the statement, let
  \(C=\ell_{1} \lor \dots\lor \ell_{t}\). We need to show we can
  derive clause $D_{i}=\lnot{\ell_{i}} \lor \lnot{b}$ for every
  $1 \leq i \leq t$.
  Assume that we derived \(D_{1},\ldots,D_{i-1}\), we derive \(D_{i}\)
  in \CLPR using witnessing assignment
  $\sigma_i:=\{\ell_{i} \mapsto 1, b \mapsto \lfalse\}$.
  Since  \(\lnot{D_{i}}=\{\ell_{i} \mapsto 1, b \mapsto 1\}\), the cost
  condition (item~\ref{item:CSRdef2} of Definition~\ref{def:CSR})
  is satisfied.
  To show the redundancy condition (item~\ref{item:CSRdef1} of
  Definition~\ref{def:CSR}) we prove something stronger, that is
  \begin{equation*}
    \Gamma\!\upharpoonright_{\lnot{D_{i}}}\ \supseteq (\Gamma\cup \{D_{1}, \ldots, D_{i}\})\!\!\upharpoonright_{\sigma_i}\;.
  \end{equation*}
  Indeed, on clauses of $\Gamma$ that do not contain the variable $b$,
  the assignments $\lnot{D_{i}}$ and $\sigma_i$ behave
  identically, while all the clauses containing \(b\)  are satisfied by \(\sigma_i\) since those clauses are \(\{C \lor b\}\) and
  \(D_{1}, \ldots, D_{i}\).
\end{proof}

\begin{proof}[Proof of \cref{thm:PHP-mn}]
  The proof is by induction. The goal is to reduce the formula to \(m-1\) pigeons
  and \(n-1\) holes. First we do some preprocessing: from \(\BPHP{m}{n}\) we derive a slightly more structured formula
  \(\FF{m}{n}\). Then we show how to derive \(\FF{m-1}{n-1}\) in
  a polynomial number of steps. The results follow because after
  \(n\) such derivations we obtain the formula \(\FF{m-n}{0}\) that contains the
clauses \(b_{1},\dots,b_{m-n}\) together with the clauses  \(\lnot{b_{m-n+1}}, \dots, \lnot{b_m}\) and $\lnot{b_{i,k,j}}$ for $i < k$ and $j$.

  We derive \(\FF{m-1}{n-1}\)  from \(\FF{m}{n}\) using the rules of \CSR calculus. We divide the
  argument into several steps, but first we show how to derive $\FF{m}{n}$ from $\BPHP{m}{n}$.

\newcounter{step}
\setcounter{step}{1}
\newcommand{\STEP}[2][]{\smallskip
{\textbf{Step \arabic{step}#1.}\addtocounter{step}{1}\itshape``#2''.\ }}

\textbf{Preprocessing 1.}{\itshape ``Make $b_i$ full-fledged extension
  variables''.} 
 We turn each variable~$b_i$ into
full-fledged extension variable that satisfy
$b_i \liff \neg ( p_{i,1} \lor \cdots \lor p_{i, n} )$, by adding the
clauses
 \[
 \ExtVars_i=\{\lnot{p_{i,j}} \lor \lnot{b_i} : j\in[n]\}
 \] one by one in \CLPR, using \cref{lem:extension_vars}.
Each blocking variable only occurs once in $\BPHP{m}{n}\cup \bigcup_{j<i}\ExtVars_j$, 
so we can apply \cref{lem:extension_vars} again to add $\ExtVars_i$. In the end we get $\Gamma_0=\BPHP{m}{n}\cup \bigcup_{i\in [m]}\ExtVars_i$. 

\textbf{Preprocessing 2.}{\itshape ``Enforce injectivity''.\ }
Optimal assignments for
\(\BPHP{m}{n}\) can have unassigned pigeons or have collisions between
pigeons.
It is more convenient to avoid collisions and just focus on
assignments that are partial matching. A moment's thought suffices to
realize that such restriction does not change the optimal cost but
simplifies the solution space.
We enforce collisions to never occur by
deriving all the unit clauses \(\lnot{b_{i,k,j}}\) by \CPR.\@ 

These clauses can be
derived in any particular order: to show that a specific \(\lnot{b_{i,k,j}}\) is
\CPR \wrt $\Gamma_0$ plus the previously derived unit clauses $\lnot{b_{i',k',j'}}$, we pick one of the two pigeons involved (say \(k\)) and use
\( \sigma=\{b_{i,k,j}= 0, b_{k}= 1, p_{k,1}=\dots=p_{k,n}=0\}\) as the
witnessing assignment.
The cost is not increased, and to check the redundancy condition
observe that \(\sigma\) satisfies all the clauses that touches in \(\Gamma_0\), and 
does not touch any previously derived unit $\lnot{b_{i',k',j'}}$.
These latter unit clauses are on both sides of the redundancy condition,
and the rest of the right side is a subset of
\(\Gamma_0\), with no occurrences of that variable \(b_{i,k,j}\). The left side has 
the same set of clauses, but restricted with \(b_{i,k,j}=0\).

Now that we have all clauses \(\lnot{b_{i,k,j}}\) we resolve them with
the corresponding clauses
$ \lnot{p_{i,j}} \lor \lnot{p_{k,j}} \lor b_{i,k,j}$ to get the set of
clauses
\[
\Inj=\{\lnot{p_{i,j}} \lor \lnot{p_{k,j}} : 1\leq i < k \leq m \text{ and }j\in [n]\}\ ,
\]
for all holes
\(j\) and pair of pigeons \(i\) and \(k\).

We do not need variables \(b_{i,k,j}\) anymore. By one application of
\cref{lmm:cleanup}, from now on we can ignore all clauses
$ \lnot{p_{i,j}} \lor \lnot{p_{k,j}} \lor b_{i,k,j}$. By another
application, we can also ignore the clauses \(\lnot{b_{i,k,j}}\).
For clarity we list the current database of clauses, where we will do induction on. 
\begin{center}
  Formula \(\FF{m}{n}\quad\)
\begin{tabular}{ll}
$ \bigvee_{j\in [n]}p_{i,j}  \lor b_i $
   & for $i \in [m]$ \hfill(\emph{totality 1}),\\[1ex]
$  \lnot{p_{i,j}}  \lor \lnot{b_i} $
   & for $i \in [m]$ and $j\in [n]$\hfill(\emph{totality 2}),\\[1ex]
$ \lnot{p_{i,j}} \lor \lnot{p_{k,j}}$
   & for $1\le i < k \le m$ and $j\in [n]$\quad \hfill(\emph{injectivity}).
\end{tabular}
\end{center}

The core idea of the induction is that if a pigeon flies to a hole, we
can assume without loss of generality that it is pigeon \(m\) that
flies into hole \(n\).

\STEP{If some pigeon $i$ flies, we can assume it is pigeon $m$ who flies}
We want to derive, in this order, the set of clauses
\[
\BB=\{\lnot{b_{m}}\, \lor b_{1},\ \lnot{b_{m}}\, \lor b_{2},\ 
\ldots,\ 
\lnot{b_{m}}\, \lor b_{(m-1)} \}
\]
from \(\FF{m}{n}\), to claim that if some pigeon is mapped, then pigeon
\(m\) is mapped too.
For each \(C_{i} = \lnot{b_{m}} \lor b_{i}\) we apply \cref{lem:CSR-symm}
using  as the witnessing permutation $\pi_{i}$, the permutation that swaps pigeons $m$ and $i$.

Namely,
\(\pi_i(p_{m,j}) = p_{i,j}\),
\(\pi_i(p_{i,j}) = p_{m,j}\),
\(\pi_i(b_m) = b_i\),
\(\pi_i(b_i) = b_m\),
and \(\pi_{i}\) is the identity on all other variables, therefore $\pi_i$ satisfies the first requirement for the lemma. Likewise  \(\lnot{C_{i}} \circ \pi_{i} \models C_{i} \), and we need to check that
\begin{equation*}
  (\FF{m}{n} \cup \{C_{1},\ldots,C_{(i-1)} \})\!\!\upharpoonright_{\assign{C_{i}}}
  \ \supseteq
  (\FF{m}{n} \cup \{C_{1},\ldots,C_{(i-1)} \})\!\!\upharpoonright_{\assign{C_{i}}\,\circ\,\pi_{i}} \;.
\end{equation*}
By symmetry
\(\FF{m}{n}\!\upharpoonright_{\lnot{C_{i}}} =
\FF{m}{n}\!\upharpoonright_{\lnot{C_{i}}\,\circ\,\pi_{i}}\), and for \(1 \leq i' < i\),
\(C_{i'}\!\upharpoonright_{\lnot{C_{i}}\,\circ\,\pi_{i}}=1\), hence the inclusion is true.
 Notice that the main reason this step works is due to the substitution $\lnot{C_{i}} \circ \pi_{i}$ setting $b_m=0$.

 The current database of clauses is $\Gamma_1=\FF{m}{n}\cup \BB$.

\STEP{If pigeon $m$ flies to some hole, we can assume it flies to hole
  $n$}
Using \CSR inferences, we derive from \(\Gamma_1\), in this order, the clauses
\[
\PP=\{\lnot{p_{m,1}}\, \lor p_{m,n},\ \lnot{p_{m,2}}\, \lor p_{m,n},\ \ldots,
\lnot{p_{m,(n-1)}}\, \lor p_{m,n} \} \]
expressing that if
pigeon \(m\) flies to some hole, this hole is the last one.

For each \(C_{j} = \lnot{p_{m,j}} \lor p_{m,n}\) we apply \cref{lem:CSR-symm}
with the witnessing permutation $\pi_{j}$~swapping holes~$n$ and~$j$. 

Namely \(\pi_j(p_{i,n}) = p_{i,j}\) and \(\pi_j(p_{i,j}) = p_{i,n}\),
and \(\pi_{j}\) is the identity on all other variables.
By construction \(\pi_{j}\) satisfies the first requirement for the
lemma, and likewise \(\lnot{C_{j}} \circ \pi_{j} \models C_{j} \), and, again, we need to check
\begin{equation*}
  (\Gamma_{1} \cup \{C_{1},\ldots,C_{(j-1)} \})\!\!\upharpoonright_{\assign{C_{j}}}
  \ \supseteq
  (\Gamma_{1} \cup \{C_{1},\ldots,C_{(j-1)} \})\!\!\upharpoonright_{\assign{C_{j}}\,\circ\,\pi_{j}} \;.
\end{equation*}
By symmetry
\(\Gamma_{1}\!\upharpoonright_{\lnot{C_{j}}} =
\Gamma_{1}\!\upharpoonright_{\lnot{C_{j}}\,\circ\,\pi_{j}}\), and for \(1 \leq j' < j\),
\(C_{j'}\!\upharpoonright_{\lnot{C_{j}}\,\circ\,\pi_{j}}=1\), hence the inclusion is true.
 Notice that the main reason this step works is due to the substitution $\lnot{C_{j}} \circ \pi_{j}$ setting $p_{m,n}=1$.
 
The current database of clauses is $\Gamma_2=\Gamma_1\cup \PP= \FF{m}{n} \cup \BB\cup \PP$.

\STEP{Obtain \(\lnot{p_{k,n}}\) for every \(1 \leq k < m \) via resolution}
Some of the proof steps in this part seem redundant and look like they could be simplified. Nevertheless we build the proof so that every intermediate clause contains the goal literal \(\lnot{p_{k,n}}\). This allows immediate application of \cref{lmm:cleanup} at the end.
Resolve the clause \(p_{m,1} \lor p_{m,2}\lor \dots\lor p_{m,n}\lor b_m\) (totality~1) with $\lnot{p_{m,n}}\lor \lnot{p_{k,n}}$ (injectivity), the resulting clause with all clauses \(\lnot{p_{m,j}} \lor p_{m,n}\) (from step 2), to get
\(b_{m} \lor p_{m,n}\lor \lnot{p_{k,n}}\).
Then resolve $b_{m} \lor p_{m,n}\lor \lnot{p_{k,n}}$ again with the clause \(\lnot{p_{m,n}} \lor \lnot{p_{k,n}}\) (injectivity), then the result with clause \(\lnot{b_{m}} \lor b_{k}\) (from step~1), and again this latter result with clause \(\lnot{b_{k}} \lor \lnot{p_{k,n}}\) (totality 2). The final result is \(\lnot{p_{k,n}}\).

The clauses $\lnot{p_{k,n}}$ subsume the clauses in $\Inj$ of the form $\lnot{p_{m,n}}\lor \lnot{p_{k,n}}$ and all the intermediate clauses from the previous resolution steps. Therefore we use \cref{lmm:cleanup} to be able to ignore the subsumed clauses.

The current database of clauses is $\Gamma_3$ is equal to 
\[
   (\FF{m}{n} \cup \BB\cup \PP \cup \{\lnot{p_{k,n}} : 1\leq k < m\}) \setminus \{\lnot{p_{m,n}} \lor \lnot{p_{k,n}} : 1\leq k < m\} .
\]

\STEP{Assign pigeon \(m\) to hole \(n\): derive unit clauses  \(p_{m,n}\) and
  \(\lnot{b_{m}}\)}
The goal is to enforce pigeon \(m\) to be mapped to hole \(n\), by deriving the clause \(p_{m,n}\) using the \CPR rule. Then we get \(\lnot{b_{m}}\) immediately by resolving
\(p_{m,n}\) with \(\lnot{p_{m,n}} \lor \lnot{b_{m}}\) (totality~2).

The unit clause $p_{m,n}$ is \CPR \wrt \(\Gamma_3\), using witness $\sigma=\{p_{m,n} = 1, b_{m} = 0\}$. 
Clearly $\sigma$~satisfies the cost condition.
To see that the redundancy condition holds as well, we need to show
that $\Gamma_{3} \!\!\upharpoonright_{\overline C}\ \vdash_1 D \!\!\upharpoonright_\sigma$ for all~$D$
in~$\Gamma_3$ that contain either $\lnot{p_{m,n}}$ or $b_m$. 
The only such clauses are
\(\lnot{p_{m,n}} \lor \lnot{b_{m}}\) and \(p_{m,1} \lor p_{m,2}\lor \dots\lor p_{m,n}\lor b_m\). They are both satisfied by
\(\sigma\).
The current database of clauses is $\Gamma_4=\Gamma_3\cup \{p_{m,n},\lnot{b_m}\}$.

\STEP{Derive $\lnot{p_{m,1}},\dots, \lnot{p_{m,(n-1)}}$ by \CPR}
We can derive them in sequence using the same witnessing substitution 
$\sigma$ setting $p_{m,n}=1$, $p_{m,1}=\dots=p_{m,(n-1)}=0$, and $b_m=0$. 
The cost condition is immediate, and the redundancy condition follows from the fact that for $1 \leq j<n$, $\sigma$ satisfies $p_{m,j}$ and
\[
(\Gamma_4 \cup \{\lnot{p_{m,1}},\dots, \lnot{p_{m,(j-1)}}\})
\!\!\upharpoonright_{p_{m,j}=1}
\ \supseteq 
(\Gamma_4 \cup \{\lnot{p_{m,1}},\dots, \lnot{p_{m,(j-1)}}\}) 
\!\!\upharpoonright_\sigma
\]
holds because all clauses touched by $p_{m,j}=1$ on the left, are satisfied by $\sigma$ on the right.

Now that we have all unit literals pigeon $m$, hole $n$ and blocking literal $b_m$, using \cref{lmm:cleanup} to remove subsumed clauses, the current database $\Gamma_5$ is

\begin{align}
 & 
 p_{i,1}\lor \cdots \lor p_{i,n-1}\lor p_{i,n}
 \lor b_i 
   &\text{ for } i \in [m-1],\label{eq:old_totality}\\[1ex]
 & \lnot{p_{i,j}}  \lor \lnot{b_i} 
   & \text{ for } i \in [m-1] \text{ and } j\in [n-1],\\[1ex]
& \lnot{p_{i,j}} \lor \lnot{p_{k,j}}
   & \text{ for } 1\le i < k \le (m-1) \text{ and } j\in [n-1],
   \\[1ex]
&p_{m,n}, \ \lnot{b_m},\ \lnot{p_{i,n}},\ \lnot{p_{m,j}} & \text{ for } i \in [m-1] \text{ and } j \in [n-1].
\end{align}

\STEP{Reduction to \(m-1\) pigeons and \(n-1\) holes}
Starting from database $\Gamma_5$ we get the (totality 1) clauses of $\FF{m-1}{n-1}$, by resolving \eqref{eq:old_totality} with units $\lnot{p_{i,n}}$. 
Now, the current database is $\FF{m-1}{n-1}$ plus the unit clauses involving hole $n$ and pigeon $m$, and the clauses  \eqref{eq:old_totality} that are now subsumed. 
We can ignore the latter two groups of clauses using \cref{lmm:cleanup}, and carry on the derivation having only $\FF{m-1}{n-1}$ in the database.

\smallskip
Finally Steps 1--6 are repeated $n-1$ times, up to derive \(\FF{m-n}{0}\).
The unit clauses derived in the whole process include
\begin{itemize}
  \item $b_{1},\dots, b_{m-n}$  (totality clauses in \(\FF{m-n}{0}\)).
  \item $\lnot{b_{m-n+1}},\dots,\lnot{b_{m}}$ (derived at each step of the induction),
  \item  $\lnot{b_{i,k,j}}$ for all $i<k$ and~$j$ (derived at the preprocessing).
\end{itemize}
Therefore $\cost{\BPHP{m}{n}}=m-n$.
\end{proof}

\section{MaxSAT Resolution \Plus Redundancy}
\label{sec:maxsat_resolution_redundancy}

We conclude this article showing a strengthening of MaxSAT
resolution, based on the redundancy rules and calculi from
\cref{sec:redundancy-rules}. For simplicity, we describe it for MaxSAT instances of the form $F=H\land S$ with hard clauses $H=\{C_1\lor b_1,\dots,C_m\lor b_m\}$ and soft clauses $S=\{\lnot{b_1},\dots,\lnot{b_m}\}$.

In \emph{MaxSAT resolution}~\cite{BLM.07} a proof starts from a set
$H$ of hard clauses and a multiset of soft clauses \(S\).

A proof starts with $H_0=H$, $S_0=S$. At
each step \(i\) new sets \(H_i,S_{i}\) are derived according to some
rules that keep invariant the minimum number of clauses in \(S_{i}\)
that must be falsified for any truth assignment satisfying $H_0$.
If at some step there are \(k\) copies of \(\bot\) among the soft
clauses, it means that at least \(k\) clauses among $C_1,\dots,C_m$ must be falsified.

Since the resolution rule is not valid for MaxSAT, we use the rules of MaxSAT resolution.
There are several equivalent ways to describe them, and we use those from~\cite{AL.23}, first used in the context of
MaxSAT in~\cite{BL.20}.

A MaxSAT resolution derivation is a sequence of pairs $(H_i,S_i)$
such that
\begin{enumerate}[label=(\alph*)]
  \item \(H_{i}:= H_{{i-1}} \cup \{C \}\), where \(C\) is deduced by
  resolution from \(H_{i-1}\) (resolution on hard clauses) and $S_i=S_{i-1}$;
  \item \(S_{i}:=S_{i-1} \cup \{C\}\) where
  \(C \in H_{i-1}\) (i.e. adding a hard clause as a soft one) and $H_i=H_{i-1}$;
  \item \(S_{i}:=S_{i-1}\setminus\{C\} \cup \{C \lor x,
  C \lor \lnot{x}\} \) where \(C \in S_{i-i}\) (split of a soft clause) and $H_i=H_{i-1}$;
  \item \(S_{i}:=S_{i-1}\setminus\{C \lor x,C \lor \lnot{x}\} \cup \{C\} \) where
\(\{C \lor x,C \lor \lnot{x}\} \subseteq S_{i-i}\) (symmetric resolution on soft clauses) and $H_i=H_{i-1}$.
\end{enumerate}

Soundness and completeness of MaxSAT
resolution are proved in~\cite{BLM.07}. In other words, if the derivation obtains some $S_i$ containing \(k\) copies of $\bot$, then among the clause $C_1,\dots, C_m$ at least $k$ of them must be falsified. This system is complete as well, in the sense that from
an instance $H\land S$ such that at least $k$ among the $C_1,\dots, C_m$ must be falsified, there is always a derivation of
\(k\) copies of \(\bot\) inside some $S_i$. 

As a deduction system on the clauses $C_1,\dots, C_m$, MaxSAT resolution is simulated by resolution~\cite[Theorem 17]{BLM.07}.
Therefore, to build inference systems stronger than MaxSAT resolution, we introduce redundancy rules
(\cref{def:CSR}).

\begin{defi}[MaxSAT resolution \Plus \CSR]\label{def:maxsat_redundancy}
Let $F=H\land S$ be a MaxSAT instance with hard clauses $H=\{C_1\lor b_1,\dots,C_m\lor b_m\}$ and soft clauses $S=\{\lnot{b_1},\dots,\lnot{b_m}\}$.	
A derivation in the system \emph{MaxSAT resolution $+$ \CSR} is a sequence of pairs of multisets of clauses $(H_i,S_i)$, where $H_i$ and $S_i$ are obtained from $H_{i-1}$ and $S_{i-1}$ either using one of the MaxSAT resolution rules (a)--(d) above or
\begin{enumerate}
\item[(a')] $H_i:=H_{i-1}\cup \{C\}$, where $C$ is \CSR w.r.t. $H_{i-1}$ (cost-redundancy on hard clauses) and $S_i=S_{i-1}$.
\end{enumerate}
Analogous systems can be defined strengthening MaxSAT resolution with
any of \CPR, \CSPR, \CLPR, \CBC.\@ 
\end{defi}

We give the definition of soundness for the system MaxSAT resolution + $P$, where $P$ is any of the rules \CPR/\CSPR/\CLPR/\CBC.

\begin{defi}[soundness]
Let $F=H\land S$ be a MaxSAT instance with hard clauses $H=\{C_1\lor b_1,\dots,C_m\lor b_m\}$ and soft clauses $S=\{\lnot{b_1},\dots,\lnot{b_m}\}$.	
 MaxSAT resolution + $P$ is \emph{sound} if whenever MaxSAT
  resolution \Plus $P$  derives $\bot$ as a soft clause with
  multiplicity $k$ from $F$, then the number of clauses that need to be falsified among $C_1,\dots,C_m$ is at least $k$, or equivalently $\cost{H}\geq k$.
  \end{defi}

It is easier to prove the soundness of MaxSAT resolution + \CSR if
we first put the derivations in a convenient normal form.

\begin{prop}[normal form]
\label{thm:normal-form}
In a MaxSAT resolution + \CSR derivation we can always assume, at the
cost of a polynomial increase in size of the derivations, that the
rules (a) and (a') are applied before the rules (b)--(d).
\end{prop}
\begin{proof}
  Having more hard clauses in the database never prevents applications
  of the rules (b)--(d), and furthermore rules (a) and (a') do not
  depend in any way on the soft clauses in it. Therefore we can assume
  all the hard clauses to be derived before any manipulation of the
  soft clauses.
\end{proof}

\begin{thm}[label=thm:soundness-max_sat_extended,name={soundness}]
  The system MaxSAT resolution + \CSR is sound.
\end{thm}

 \begin{proof}
  Consider a MaxSAT resolution + \CSR derivation
  $(H_0,S_0),\dots,(H_t,S_t)$ where \(S_{t}\) contains $k$ copies of
  \(\bot\).

  Given a total assignment $\alpha$, let $c(\alpha,S_i)$ be the number
  (counted with multiplicity) of clauses in $S_i$ falsified by
  $\alpha$. That is, in particular, for any total assignment $\alpha$,
  $c(\alpha,S_t)\geq k$.
  By induction on $i$, we will show that
  \begin{equation}\label{eq:stepi}
    \min_{\alpha\;:\;\alpha\;\models\;H_i} c(\alpha,S_i) = \cost{H_0}\;.
  \end{equation}
  Because of \cref{thm:normal-form} we can assume that up to
  some step \(t_{0}\), included, the proof only uses rules (a) or
  (a'), and from step \(t_{0}+1\) to \(t\), the proof only uses rules
  (b)--(d).

  By the soundess of \(\CSR\) (\cref{lem:CSR-redundant}) we have that when \(i \leq t_{0}\),
  \(H_{i}\) is satisfiable and \(\cost{H_{i}}=\cost{H_{0}}\).
  Furthermore \(S_{i}=S_{0}\). Hence
  \begin{equation*}
    \min_{\alpha\;:\;\alpha\;\models\;H_i} c(\alpha,S_i) =
    \min_{\alpha\;:\;\alpha\;\models\;H_i} c(\alpha,S_0) =
    \cost{H_i} = \cost{H_0}\;.
  \end{equation*}
  for \(i \leq t_{0}\). For \(i>t_{0}\) we applied one of rules
  (b)--(d) above, and
  \[
	\min_{\alpha\;:\;\alpha\;\models\;H_i}c(\alpha,S_i)=\min_{\alpha\;:\;\alpha\;\models\;H_{i-1}}c(\alpha,S_{i-1})\stackrel{IH}{=}\cost{H_0}\ ,
  \]
  where the first equality is the soundness of MaxSAT resolution.
\end{proof}

We give the definition of completeness for the system MaxSAT resolution + $P$.

\begin{defi}[completeness]
Let $F=H\land S$ be a MaxSAT instance with hard clauses $H=\{C_1\lor b_1,\dots,C_m\lor b_m\}$ and soft clauses $S=\{\lnot{b_1},\dots,\lnot{b_m}\}$.	
 MaxSAT resolution + $P$ is \emph{complete} if 
 the number of clauses that need to be falsified among $C_1,\dots,C_m$ is at least $k$, then there is a MaxSAT
  resolution \Plus $P$  derivation of $\bot$ as a soft clause with
  multiplicity $k$ from $F$.
  \end{defi}

\cref{def:maxsat_redundancy} extends the systems from
\cref{sec:redundancy-rules}, but those systems formally prove
that cost is at least \(k\) by deriving \(k\) hard unit clauses
\(b_{i_{1}}, b_{i_{2}}, \ldots, b_{i_{k}}\), while here we want to
prove $k$ copies of $\bot$ instead.
This is not an issue since we can copy these
$b_{i_{1}}, b_{i_{2}}, \ldots, b_{i_{k}}$ in the database of soft
clauses, using rule~(b), and resolve them with the corresponding soft
clauses
\(\lnot{b_{i_{1}}}, \lnot{b_{i_{2}}}, \ldots, \lnot{b_{i_{k}}}\),
using rule~(d), to get $k$ empty clauses $\bot$.
That is, MaxSAT resolution + $P$ p-simulates the original
calculus $P$ and therefore is also complete when $P$ is either \CSPR,
\CPR, or \CSR (\cref{thm:SPR-completeness}).
The completeness of MaxSAT resolution + \CBC and \CLPR follows instead
from the completeness of MaxSAT resolution itself~\cite{BLM.07}.

\begin{thm}[completeness]
  Let $P$ denote any of \CSR, \CPR, \CSPR, \CLPR, \CBC.\@ The system
  MaxSAT resolution \Plus $P$ is complete.
\end{thm}

\section{Conclusions and open problems}
\label{sec:conclusions}
We proposed a way to extend redundancy rules, originally introduced for
SAT, into polynomially verifiable rules for MaxSAT.
We defined sound and complete calculi based on those rules and we
showed the strength of some of the calculi giving short derivations of
notable principles and we showed the incompleteness of the weaker ones and width lower bounds for the stronger ones.
We conclude this article with a list of open problems:
\begin{enumerate}
  \item The cost constraint for the redundancy rules is very strict, for example compared to the rule \texttt{CPR} in \cite{IBJ.22}. Indeed, \CPR enforces the check on the cost even on assignments falsifying the hard clauses of the formula.
   Is it possible to relax \CPR
  without giving up on efficient verification as in \cite{IBJ.22}?
  \item Does \CSR simulate MaxSAT Resolution? That is, if we have
  a MaxSAT instance $\Gamma$ with blocking variables and MaxSAT
  Resolution proves in size $s$ that $\cost{\Gamma}=k$, is there
  a proof of $\cost{\Gamma}=k$ in \CSR of size $\mathsf{poly}(s)$?
  \item We proved a width lower bound for \CSPR and an analogue of a width lower bound for \CSR on formulas with optimal assignments far from each other in the Hamming distance. We reiterate the open problem of proving size lower bounds for \CSPR and stronger systems.
\end{enumerate}

\section*{Acknowledgment}
\noindent The authors would like to thank the Simons Institute for the Theory of Computing: part of this work has been done during the \emph{Extended Reunion: Satisfiability} program (Spring 2023). Another part of this work has been done during the Oberwolfach workshop 2413 \emph{Proof Complexity and Beyond} and during the 2023 Workshop on Proof Theory and its Applications organized by the Proof Society. We would also like to thank the anonymous reviewers for their careful reading of the manuscript and for their constructive comments.

\newcommand{\etalchar}[1]{$^{#1}$}

\end{document}